\documentclass[12pt,reqno]{amsart}
\usepackage{amsmath, amssymb, amsthm, amsfonts}

\usepackage{bbm}
\usepackage{mathrsfs} \usepackage{upgreek}
\usepackage{amscd}
\usepackage{graphicx}
\usepackage[all]{xy}
\usepackage{cite}
\usepackage{tikz}
 \usepackage{color} \usepackage{tikz-cd}
 \usepackage{hyperref}
\usepackage[english]{babel}
\usepackage[margin=1.0in]{geometry}
\usepackage{amsmath, amssymb, amsthm, amsfonts}

\usepackage{mathrsfs} \usepackage{upgreek}
\usepackage{amscd}
\usepackage{graphicx}
\usepackage[all]{xy}
\usepackage{cite}
\usepackage{tikz}
 \usepackage{color} \usepackage{tikz-cd}
 \usepackage{hyperref}
\usepackage[english]{babel}
\usepackage[margin=1.0in]{geometry}

\DeclareMathOperator{\Res}{Res} 
\DeclareMathOperator{\ch}{ch}

\newtheorem{theorem}{Theorem}[section]
\newtheorem{corollary}[theorem]{Corollary}
\newtheorem{lemma}[theorem]{Lemma}

\newtheorem{proposition}[theorem]{Proposition}

\newtheorem{remark}[theorem]{Remark}
\newtheorem{conjecture}[theorem]{Conjecture}

\begin{document}

\title{Virasoro constraints for Hodge integrals}
\author{Xin Wang}
\address{Department of Mathematics\\ Shandong University \\ Jinan, China}

\email{xinwmath@gmail.com}
\begin{abstract}
The purpose of this paper is to study Virasoro constraints  for Hodge integrals in Gromov-Witten theory of any target varieties. Results consist of the following: Firstly, we propose Virasoro conjecture for 
Hodge integrals in Gromov-Witten theory of any target varieties; Secondly, we prove Virasoro constraints for Hodge integrals  in genus zero of any target varieties; Thirdly, we prove the genus-1 $L_1^{\mathbb{E}}$ constraint with one Hodge character class insertion for any target varieties;   Lastly, we obtain certain new vanishing identities  in higher genus for Gromov-Witten invariants of any target varieties. 
\end{abstract}

\maketitle
\tableofcontents

\allowdisplaybreaks

\section{Introduction}
% \subsection{Hodge integral}
 The Virasoro conjecture predicts that the generating function of descendant Gromov-Witten invariants for smooth projective varieties is annihilated by a sequence of differential operators which form a half branch of Virasoro algebra.  This conjecture was proposed by Eguchi, Hori and Xiong (cf. \cite{MR1454328})
and also by Katz  (cf. \cite{MR1635867}, \cite{MR1677117}). When the manifold is a point, the  Virasoro conjecture is equivalent to Witten's conjecture on the intersection numbers of moduli space of curves (cf. \cite{MR1144529}),  which was firstly proved by Kontsevich (cf. \cite{MR1171758}).

 Hodge integrals over moduli space of stable maps was introduced in \cite{MR1728879}.  It was proven in\cite{MR1728879}  
that Hodge integrals can be reconstructed from pure descendant   Gromov–Witten invariants. Furthermore it was pointed by Givental that this reconstruction can be realized by  quantizition of a symplectic operator (cf. \cite{MR1901075}). When the target mainfold is a point, Hodge integrals over the moduli space of curves have been well studied in many works and play an important role in localization computations of Gromov-Witten invariants of toric varieties.

Virasoro conjecture is a long time problem in Gromov-Witten theory. 
Since Hodge integrals are certain generalization of usual Gromov-Witten invariants, people expected there should be a similar Virasoro constraints for Hodge integrals. For the target manifold is a point, the Virasoro constraints for Hodge integrals over the moduli space of curves were written explicitly in \cite{MR3623280}, \cite{MR4466672} and \cite{zhouHodge}.

Let $X$ be a smooth projective variety of complex dimension $d$.  For simplicity, we assume that $H^{odd}(X;\mathbb{C})=0$. Fix a basis $\{\phi_1,...,\phi_N\}$ of $H^*(X;\mathbb{C})$
 with $\phi_1$ equal to the identity of the cohomology ring of $X$ and  $\phi_\alpha\in H^{p_\alpha,q_\alpha}(X;\mathbb{C})$
for every $\alpha$.  Let $\mathcal{C}=(\mathcal{C}_\alpha^\beta)$ be the matrix of
multiplication by the first Chern class $c_1(X)$ in the ordinary cohomology ring, i.e.
$c_1(X)\cup \phi_\alpha=\sum_{\beta}\mathcal{C}_\alpha^\beta\phi_\beta$. Let $F_g^{\mathbb{E}}(\mathbf{t},\mathbf{s})$ be the generating function of genus-$g$ Hodge integrals over moduli
spaces of stable maps with target $X$. Precise definition of $F_g^{\mathbb{E}}(\mathbf{t},\mathbf{s})$ will be given in
equation~\eqref{eqn:def0FgE}. Set $\mathcal{D}^{\mathbb{E}}(\mathbf{t},\mathbf{s})=\exp(\sum_{g\geq0}\hbar^{2g-2}F_g^{\mathbb{E}}(\mathbf{t},\mathbf{s}))$ be the all genus total potential function of Hodge integrals, where $\hbar$ is a formal parameter. 
For $n\geq-1$, we introduce following operators  
\begin{align*}
L_n^{\mathbb{E}}
=&\sum_{m=0}^{n+1}\frac{(-1)^{m}}{m!}\sum_{k_1,...,k_m\ge1}\left(\prod_{i=1}^{m}\left(\frac{B_{2k_i}}{(2k_i)!}\mathfrak{s}_{k_i}\right)\right)
\\&\hspace{15pt}\cdot\left\{\sum_{j=0}^{n+1}\sum_{r,\alpha,\beta}\left(\left(\prod_{i=1}^{m}\Delta_{2k_i-1}\right)\hat{e}_{n+1-j}(n,\alpha)\right)(r) (\mathcal{C}^j)_\alpha^\beta\tilde{t}_r^\alpha \frac{\partial}{\partial t_{\sum_{i=1}^{m}(2k_i-1)+r+n-j}^{\beta}}
\right.\\&\left.\hspace{20pt}-\frac{\hbar^2}{2}\sum_{j=0}^{n+1}\sum_{s,\alpha,\beta}(-1)^{-s}\left(\left(\prod_{i=1}^{m}\Delta_{2k_i-1}\right)\check{e}_{n+1-j}(n,\alpha)\right)(-s-1)(\mathcal{C}^j)^{\alpha\beta}\frac{\partial}{\partial t_s^\alpha}
\frac{\partial}{\partial t_{\sum_{i=1}^{m}(2k_i-1)-s-1+n-j}^\beta}\right\}
\\&+\frac{1}{2\hbar^2}\sum_{\alpha,\beta}(\mathcal{C}^{n+1})_{\alpha\beta}t_0^\alpha t_0^\beta
+\frac{\delta_{n,0}}{24}
\int_{X}\left(\frac{3-d}{2}c_d(X)-c_1(X)\cup c_{d-1}(X)\right)+\frac{\delta_{n,-1}}{24}\mathfrak{s}_{1}\int_{X}c_d(X),\quad n\geq-1.
\end{align*}
where $\Delta_{k}$ is the difference operator defined by
\[(\Delta_{k}f)(x):=f(x+k)-f(x)\]
and two functions $\hat{e}_j(n,\alpha)$ and $\check{e}_j(n,\alpha)$ are the $j$-th elementary symmetric functions
\[\hat{e}_j(n,\alpha)(x)=e_j(x+b_\alpha,...,x+n+b_\alpha),\quad\check{e}_j(n,\alpha)(x)=e_j(x+b^\alpha,...,x+n+b^\alpha)\]
with $b_\alpha:=p_\alpha-\frac{1}{2}(d-1)$ and $b^\alpha:=1-b_\alpha$. 
Here $\mathcal{C}^j$
is the $j$-th power of the matrix $\mathcal{C}$, $(\mathcal{C}^j)^{\alpha\beta}$ are entries of the matrix $(\mathcal{C}^{j})\eta^{-1}$ and $(\mathcal{C}^{n+1})_{\alpha\beta}$ are entries of the matrix $(\mathcal{C}^{n+1})\eta$. For the case when $X$ is a point  the same expressions were found in \cite{zhouHodge}.
These operators $\{L_n^{\mathbb{E}}\}_{n\geq-1}$ satisfy Virasoro bracket relation
\[[L_n^{\mathbb{E}},L_m^{\mathbb{E}}]=(n-m)L_{n+m}^{\mathbb{E}}.\]

Now we propose 
Virasoro conjecture for Hodge integrals  as follows:
\begin{conjecture}\label{conj:Hodge}
For any smooth projective variety, 
\[L_n^{\mathbb{E}}\mathcal{D}^{\mathbb{E}}(\mathbf{t},\mathbf{s})=0, \quad n\geq-1.\]
\end{conjecture}

Motived by  Givental-Teleman reconstruction results for semisimple cohomological field theory (cf. \cite{MR2917177}), we have the following
\begin{theorem}\label{thm:hodge-semisimple}
For any smooth projective variety with semisimple quantum cohomology, Virasoro conjecture~\ref{conj:Hodge} for Hodge integrals holds.
% \[\Psi_{g,n;k_1,...k_m}^{\mathbb{E}}(\mathbf{t})=0\]
% for any $g,n, k_1,...,k_m$.
\end{theorem}
In \cite{MR3474326},  an algorithm was proposed to solve the differential equations satisfied by the Hodge
potentials associated to an arbitrary semisimple Frobenius manifold. This algorithm
 represents the Hodge potential in terms of the genus zero generating function the
Frobenius manifold and the genus zero two-point functions, and shows that the Hodge
potential is the logarithm of a tau-function of an integrable hierarchy of Hamiltonian
evolutionary PDEs called the Hodge hierarchy, which is a tau-symmetric integrable
deformation of the principal hierarchy of the Frobenius manifold with deformation
parameters $\mathbf{s}=\{\mathfrak{s}_k\}_{ k\geq 1}$ and $\hbar$. We hope Virasoro constraints for Hodge integrals in semisimple cases would provide us a new way to construct the Hodge integrable hierachy. 

The Virasoro conjecture for Gromov-Witten invariants of target varieties with non-semisimple quantum cohomology is widely open. It is very interesting to investigate to what extend we could prove the Virasoro conjecture  for Hodge integrals of general target varieties.  
Let $\Psi_{g,n}^{\mathbb{E}}(\mathbf{t},\mathbf{s})$ be the coefficient of $\hbar^{2g-2}$
in the Laurent expansion of $\mathcal{D}^{\mathbb{E}}(\mathbf{t},\mathbf{s})^{-1}L_n^{\mathbb{E}}\mathcal{D}^{\mathbb{E}}(\mathbf{t},\mathbf{s})$. That is
\[\mathcal{D}^{\mathbb{E}}(\mathbf{t},\mathbf{s})^{-1}L_n^{\mathbb{E}}\mathcal{D}^{\mathbb{E}}(\mathbf{t},\mathbf{s})
=\sum_{g\geq0}\hbar^{2g-2}\Psi_{g,n}^{\mathbb{E}}(\mathbf{t},\mathbf{s}).\]
Moreover, we expand $\Psi_{g,n}^{\mathbb{E}}(\mathbf{t},\mathbf{s})$ as a formal power series of $\mathbf{s}$
\[\Psi_{g,n}^{\mathbb{E}}(\mathbf{t},\mathbf{s})
=\sum_{m\geq0}\sum_{k_1,...,k_m\geq0}
\frac{1}{m!}\mathfrak{s}_{k_1}...\mathfrak{s}_{k_m}\Psi_{g,n;k_1,...k_m}^{\mathbb{E}}(\mathbf{t}).\]
We call the equation $L_n^{\mathbb{E}}\mathcal{D}^{\mathbb{E}}=0$ the $L_n^{\mathbb{E}}$-constraint for the partition function of Hodge integrals. It is equivalent
to $\Psi^{\mathbb{E}}_{g,n;k_1,...,k_m}=0$ for all $g$ and $k_1,...,k_m$. The equation $\Psi^{\mathbb{E}}_{g,n;k_1,...,k_m}=0$ will be called  genus-$g$ degree-$(k_1,...,k_m)$ $L_n^{\mathbb{E}}$-constraint. In particular, we denote $\Psi_{g,n;\emptyset}$ to be the function $\Psi_{g,n}^{\mathbb{E}}(\mathbf{t},\mathbf{s})|_{\mathbf{s}=0}$.

\begin{theorem}\label{thm:g=0-L_n^E-constraints}
For any smooth projective variety, the genus-0 $L_n^{\mathbb{E}}$ constraints hold.    
\end{theorem}
From this theorem, we will see Virasoro constraints for Hodge integrals can be seen as an all genus extension  of Virasoro constraints and genus-0 $\widetilde{L}_n$ constraints for Gromov-Witten invariants proposed by Eguchi, Hori and Xiong in \cite{MR1454328}.  

It is well known that the vanishing of $\Psi_{1,n;\emptyset}$ is still open for general target varieties. In \cite{MR1815723}, it was proved that the vanishing of $\Psi_{1,n;\emptyset}$ for all $n$ can be reduced to vanishing of $\Psi_{1,1;\emptyset}$ (cf. \cite{MR1815723}). It seems tautological relations on the genus-1 moduli space of curves is not sufficient  to imply the vanishing of $\Psi_{1,1;\emptyset}$. However, we prove the genus-1 tautological relations imply the vanishing of $\Psi^{\mathbb{E}}_{1,1;k_1}$  in the following theorem.
\begin{theorem}\label{thm:psiEg=1}
For any smooth projective variety, genus-$1$, degree-$k_1$, $L_1^{\mathbb{E}}$ constraints  hold, 
i.e.
\[\Psi_{1,1 ;k_1}^{\mathbb{E}}=0\] for any $k_1\geq1$. 
\end{theorem}
For higher genus $g\geq2$, the similar vanishing of the genus-$g$, degree-$k_1$, $L_1^{\mathbb{E}}$ constraints can be proved if degree $k_1$ is big enough.
\begin{theorem}\label{thm:psiEg1k1}
For any smooth projective variety,
 \begin{align*}
 \Psi^{\mathbb{E}}_{g,1;k_1}=0    
 \end{align*}
 for $g\geq2, k_1\geq\max\{g+1,\frac{3g-1}{2}\}$.
\end{theorem}
% \subsection{Linear Hodge integral} 

This paper is organized as follows. In section~\ref{sec:virasoro-conj-Hodge}, we recall  the basic knowledge in Gromov-Witten theory and Hodge integral. In section~\ref{sec:Virasoro-gw-Hodge}, we give the Virasoro operators for Gromov-Witten invariants and Hodge integrals.  In section~\ref{sec:g=0-Virasoro-Hodge}, we prove the genus-0 Virasoro constraints for any target variety.   In section~\ref{sec:proof-thm-psiEg=1}, we give the proof of Theorem~\ref{thm:psiEg=1}.   In section~\ref{sec:proof-thm-g>=2}, we give the proof of Theorem~\ref{thm:psiEg1k1}. 
%  \\{\bf  Acknowledgements.}The author would like to thank the anonymous referees for careful
% reading of the manuscript and for the many helpful suggestions and comments. The author would like to thank professor Xiaobo Liu, Felix Janda, Yunfeng Jiang, Di Yang and Youjin Zhang for discussion about Hodge integrals, Gromov-Witten invariants and Frobenius manifold. This work is supported by  National Key R\&D Program of China (2022YFA1006200)  and  Shandong Provincial Natural Science Foundation (Grant No. ZR2021MA101). 

\section{Gromov-Witten
invariants and Hodge integrals}\label{sec:virasoro-conj-Hodge}
\subsection{Gromov-Witten
invariants and Hodge integrals}
Let $X$ be a smooth projective variety of complex dimension $d$, $\{\phi_1,...,\phi_N\}$ be a graded basis of
$H^*(X;\mathbb{C})$, such that $\phi_1$ is the identity element and $\phi_\alpha\in H^{p_\alpha,q_\alpha}(X;\mathbb{C})$. Let $\eta=(\eta_{\alpha\beta})$, $\eta_{\alpha\beta}:=\int_{X}\phi_\alpha\cup\phi_\beta$ be the intersection pairing on $H^*(X;\mathbb{C})$  and $\eta^{-1}=(\eta^{\alpha\beta})$ be the inverse of matrix $\eta$.
We will use $\eta$ and $\eta^{-1}$ to lower and raise indices.

Recall that the big phase space is by definition the product of infinite copies of $H^*(X;\mathbb{C})$, that is $\mathcal{P}:=\prod_{n=0}^{\infty}H^*(X;\mathbb{C})$.  Then we denote the corresponding basis for the $n$-th copy of $H^*(X;\mathbb{C})$
 in $\mathcal{P}$ by $\{\tau_{n}(\phi_\alpha)\}$.   We can think of $\mathcal{P}$ as an infinite dimensional vector space with basis $\{\tau_n(\phi_\alpha):1\leq \alpha\leq N,n\geq0\}$. Let $\{t_n^\alpha:n\geq0,\alpha=1,...,N\}$ be
the corresponding coordinate system on $\mathcal{P}$. For convenience, we identify $\tau_n(\phi_\alpha)$ with the
coordinate vector field $\frac{\partial}{\partial t_n^\alpha}$
on $\mathcal{P}$ for $n\geq0$. If $n<0$, $\tau_n(\phi_\alpha)$ is understood as the 0 vector
field. We also abbreviate $\tau_0(\phi_\alpha)$ as $\phi_\alpha$.  We use $\tau_+$ and $\tau_-$
to denote the operator which shift the level of descendants, i.e.
\[\tau_{\pm}\left(\sum_{n,\alpha}f_{n,\alpha}\tau_{n}(\phi_\alpha)\right):=\sum_{n,\alpha}f_{n,\alpha}\tau_{n\pm 1}(\phi_\alpha)\]
where $f_{n,\alpha}$ are functions on the big phase space.

Let $\overline{\mathcal{M}}_{g,n}(X,A)$ be the moduli space of stable maps $f:(C;x_1,...,x_n)\rightarrow X$,  where $(C;x_1,...,x_n)$
is a genus-$g$ nodal curve with $n$ marked points and $f_*([C])=A\in H_2(X;\mathbb{Z})$.
The descendant Gromov-Witten invariants is defined to be
\begin{align}\label{eqn:def-des-gw-inv}
\langle\tau_{k_1}(\phi_{\alpha_1})...\tau_{k_n}(\phi_{\alpha_n})\rangle_{g,n}:=\sum_{A\in H_2(X;\mathbb{Z})}q^{A}\int_{[\overline{\mathcal{M}}_{g,n}(X,A)]_{vir}} \prod_{i=1}^{n}ev_i^*\phi_{\alpha_i}\cup c_1(L_i)^{k_i}   
\end{align}
where $L_i$ are the tautological line bundles over $\overline{\mathcal{M}}_{g,n}(X,A)$, and $ev_i: \overline{\mathcal{M}}_{g,n}(X,A)\rightarrow X$ are
the evaluation maps. Let  $\mathbf{t}$ be the sets of variables $\{t_n^\alpha,n\geq0, \alpha=1,...,N\}$. 
We
define the generating
function of genus-$g$ Gromov-Witten invariants
\begin{align*}
F_g(\mathbf{t}):=\langle\exp(\sum_{n,\alpha}t_n^\alpha\tau_n(\phi_\alpha))\rangle_{g} 
\end{align*}
and its derivatives
\begin{align*}
\langle\langle\tau_{k_1}(\phi_{\alpha_1})...\tau_{k_n}(\phi_{\alpha_n})\rangle\rangle_{g}
:=\frac{\partial^n}{\partial t_{k_1}^{\alpha_1}...\partial t_{k_n}^{\alpha_n}}F_{g}(\mathbf{t}).
\end{align*}
For our purpose, we also define covariant derivatives of $F_g$ with respect to the trivial
connection on $\mathcal{P}$ by
\[\langle\langle W_1,...,W_n\rangle\rangle_g:=\sum_{k_1,\alpha_1,...,k_n,\alpha_n}f^{1}_{k_1,\alpha}...f^n_{k_n,\alpha_n}\frac{\partial^n}{\partial t_{k_1}^{\alpha_1}...\partial t_{k_n}^{\alpha_n}}F_{g}(\mathbf{t})\]
for vector fields $W_i=\sum_{k,\alpha}f_{k,\alpha}^{i}\tau_{k}(\phi_\alpha)$ on the big phase space. Using this double bracket, we define the quantum product on the big phase space: for any vector fields $W_1$ and $W_2$,
\begin{align*}
W_1\bullet W_2=\sum_{\alpha}\langle\langle W_1W_2\phi_\alpha\rangle\rangle_0\phi^\alpha.
\end{align*} 
The total descendant  potential is defined to be
\[\mathcal{D}(\mathbf{t}):=\exp\left(\sum_{g\geq0}\hbar^{2g-2}F_g(\mathbf{t})\right).\]

The
Hodge bundle $\mathbb{E}$ over $\overline{\mathcal{M}}_{g,n}(X,A)$ is the rank $g$ vector bundle with fiber $H^0(C,\omega_C)$ over the domain curve $(C;x_1,...,x_n)$, where $\omega_C$ is the dualizing sheaf of curve $C$. Let $\ch_k(E)$ be
the $k$-th Chern character of $\mathbb{E}$. By Mumford’s relations (cf. \cite{MR717614}), the $\ch_{k}(\mathbb{E})$  vanish for $k$ even and positive. 
The Hodge integral is defined to be 
\begin{align}\label{eqn-def:Hodge-cor}
\langle\tau_{k_1}(\phi_{\alpha_1})...\tau_{k_n}(\phi_{\alpha_n});\prod_{j=1}^{m}\ch_{l_j}(\mathbb{E})\rangle_{g,n}:=\sum_{A\in H_2(X;\mathbb{Z})}q^{A}\int_{[\overline{\mathcal{M}}_{g,n}(X,A)]_{vir}} \prod_{i=1}^{n}ev_i^*\phi_{\alpha_i}\cup c_1(L_i)^{k_i}   \cup \prod_{j=1}^{m}\ch_{l_j}(\mathbb{E}).   
\end{align}
By definition, the descendant Gromov-Witten invariants~\eqref{eqn:def-des-gw-inv} can be recovered from Hodge integrals~\eqref{eqn-def:Hodge-cor} by taking $m=0$.

\subsection{Faber-Pandharipande formula for Hodge integrals}

The total Hodge potential of $X$ is defined as an extension of
the total descendent potential depending on the sequence $\mathbf{s}=(\mathfrak{s}_1,\mathfrak{s}_2,...)$ of new
variables and incorporating intersection indices with characteristic classes of the
Hodge bundles. To be precise, we
define the generating
function of genus-$g$ Hodge integrals
\begin{align}\label{eqn:def0FgE}
F_g^{\mathbb{E}}(\mathbf{t},\mathbf{s}):=\langle\exp(\sum_{n,\alpha}t_n^\alpha\tau_n(\phi_\alpha));\exp(\sum_{m}\mathfrak{s}_m \ch_{2m-1}(\mathbb{E}))\rangle_{g}.    
\end{align}
The total Hodge potential is defined to be 
\[\mathcal{D}^{\mathbb{E}}(\mathbf{t},\mathbf{s}):=\exp(\sum_{g\geq0}\hbar^{2g-2}F_g^{\mathbb{E}}(\mathbf{t},\mathbf{s})).\]
For $l\geq1$, we introduce a formal differential operator
\begin{align*}
D_{2l-1}=-\frac{\partial}{\partial \mathfrak{s}_{l}}+\frac{B_{2l}}{(2l)!}\left(-\sum_{n,\alpha}\Tilde{t}_n^\alpha\frac{\partial}{\partial t_{n+2l-1}^\alpha}+\frac{\hbar^2}{2}\sum_{i,\alpha,\beta}(-1)^i\eta^{\alpha\beta}\frac{\partial}{\partial t_i^\alpha}\frac{\partial}{\partial t_{2l-2-i}^\beta}\right)    
\end{align*}
where $B_{2l}$ are the Bernoulli numbers.
In \cite{MR1728879}, Faber and Pandharipande proved that the total Hodge potential is annihilated by differential operators $D_{2l-1}$
\begin{align}\label{eqn:FP-formula}
D_{2l-1}\mathcal{D}^{\mathbb{E}}(\mathbf{t},\mathbf{s})=0,\quad l\geq1.    
\end{align}
\subsection{Givental quantization}
In this section, we review the basic concepts of the quantization of quadratic Hamiltonians (see \cite{MR1901075} for more details).

Let $z$ be a formal variable. We consider the space $\mathbb{H}$ which is the space of Laurent polynomials in one variable $z$ with coefficients in $H^*(X;\mathbb{C})$.
We define the symplectic form $\Omega$ on $\mathbb{H}$ by
\[\Omega(f,g)=\Res_{z=0}\eta( f(-z),g(z))dz\]
for any $f,g\in\mathbb{H}$, where $\Res_{z=0}$ means taking the residue at $z=0$. Note that we have $\Omega(f,g)=-\Omega(g,f)$. There is a natural polarization $\mathbb{H}=\mathbb{H}_+\oplus \mathbb{H}_-$ corresponding to the decomposition $f(z,z^{-1})=f_+(z)+f_-(z^{-1})z^{-1}$ of Laurent polynomials into polynomial and polar parts. It is easy to see that $\mathbb{H}_+$ and $\mathbb{H}_-$ are both Lagrangian subspaces of $\mathbb{H}$ with respect to $\Omega$.

Introduce a Darboux coordinate system $\{p^\beta_m,q^\alpha_n\}$ on $\mathbb{H}$ with respect to the above polarization. A general element $f\in\mathbb{H}$ can be written in the form
\[\sum_{m,\beta}p^\beta_m\phi^\beta(-z)^{-m-1}+\sum_{n,\alpha}q^\alpha_n\phi_{\alpha} z^n,\]
where $\{{\phi}^\alpha\}$ is the dual basis of $\{\phi_\alpha\}$. Denote
\begin{eqnarray*}
\mathbf{p}(z):=\sum_{m,\beta}p^\beta_m\phi^\beta(-z)^{-m-1},\quad
\mathbf{q}(z):=\sum_{n,\alpha}q^\alpha_n\phi_{\alpha} z^n.
\end{eqnarray*}
For convenience, we always identify the space $\mathbb{H}_{+}$ with the big phase space $\mathcal{P}$ by identifying the basis $\{z^n\phi_\alpha\}$ with $\{\tau_n(\phi_\alpha)\}$. 
% We introduced the formal power series $\mathbf{t}(z)=t_0+t_1z+t_2z^2+\cdots$, which appears as the insertion in the genus $g$ correlations,  with $z$ replaced by $\psi$ class.  By  the dilaton shift: $\mathbf{q}(z)=\mathbf{t}(z)-\mathbf{1} z$, we relate $\mathbf{t}(z)$ to the Darboux coordinates. The dilaton shift appears naturally in the quantization procedure. 

Let $A:\mathbb{H}\to\mathbb{H}$ be a linear infinitesimally symplectic transformation, i.e. $\Omega(Af,g)+\Omega(f,Ag)=0$ for any $f,g\in\mathbb{H}$. Under the Darboux coordinates, the quadratic Hamiltonian
\[ h_{A}: f\mapsto\frac{1}{2}\Omega(Af,f)\]
is a series of homogeneous degree two monomials in $\{p^\beta_m,q^\alpha_n\}$. Let $\hbar$ be a formal variable and define the quantization of quadratic monomials as
\begin{align*}
\widehat{q^\beta_mq_n^\alpha}=\frac{q^\beta_m q_n^\alpha}{\hbar^2}, \widehat{q^\beta_m p_n^\alpha}=q^\beta_m\frac{\partial}{\partial q^\alpha_n},
  \widehat{p^\beta_m p_n^\alpha}=\hbar^2 \frac{\partial}{\partial q^\beta_m}\frac{\partial}{\partial q^\alpha_n}.  
\end{align*}
We define the quantization $\widehat{A}$ by extending the above equalities linearly. The differential operators $\widehat{q^\beta_mq_n^\alpha},\widehat{q^\beta_mp_n^\alpha},\widehat{p^\beta_mp_n^\alpha}$ act on the so called Fock space  which is the space of formal functions in $\mathbf{q}(z)\in\mathbb{H}_+$. For example, the descendent potential $\mathcal{D}(\mathbf{t})$ is regarded as elements in  Fock space via the dilaton shift $\mathbf{q}(z)=\mathbf{t}(z)-\mathbf{1} z$. The quantization of a symplectic transform of the form $\exp(A)$, with $A$ infinitesimally symplectic, is defined to be $\exp(\widehat{A})=\sum_{n\geq 0}\frac{\widehat{A}^n}{n!}$.

The quantization procedure gives  a projective representation of the Lie algebra of
infinitesimal symplectomorphisms. For infinitesimal symplectomorphisms $F$ and $G$ we have
\begin{align*}
[\widehat{F},\widehat{G}]=\{F,G\}^{\widehat{\empty}}  +C(h_F,h_G)  
\end{align*}
where $\{\,,\,\}$ is the Lie bracket, $[\,,\, ]$ is the commutator, $h_F$ and $h_G$ is the
quadratic Hamiltonian corresponding to $F$ and $G$, and $C$ is the cocycle satisfy
\begin{align*}
C(p_\alpha^2,q_\alpha^2)=2,\quad  C(p_\alpha p_\beta,q_\alpha q_\beta)=1 \quad \text{for} \, \alpha\neq\beta,  
\end{align*}
and $C=0$ for any other pairs of quadratic Darboux monomials.

Using the above formalism, Givental rewrite Virasoro operators $L_n$ as quantization of quadratic hamitonians (cf. \cite{MR1901075}). Define $l_0=z\frac{d}{dz}+\mu+\frac{1}{2}+\frac{\rho }{z}$, where $\mu: H^*(X)\rightarrow H^*(X)$ is the Hodge grading
operator, and $\rho$ is the operator of multiplication by $c_1(X)$ using the classical cup product on $H^*(X)$. Considering a sequence of infinitesimal symplectic tranformation $\{l_n:=l_0(zl_0)^{n}\}_{n\geq0}$, then up to  a minus sign, Virasoro operators $\{L_n\}$ is exactly the quantization of quadratic hamitonians of infinitesimal symplectic transformation $\{l_n\}$. 

It is obvious that multiplication by $z^{2k-1}$ defines an infinitesimal symplectic transformation on $(\mathbb{H},\Omega)$, and we denote by $\widehat{z^{2k-1}}$ the
corresponding quantization. Givental pointed out that Faber-Panharpande formula~\eqref{eqn:FP-formula}
can be reformulated as
\begin{align}\label{eqn:FP-rel-quan}
\mathcal{D}^{\mathbb{E}}(\mathbf{t},\mathbf{s})=\exp\left(\frac{B_{2k}}{(2k)!}\sum_{k\geq1}\mathfrak{s}_{k}\widehat{z^{2k-1}}\right)\mathcal{D}(\mathbf{t}).    
\end{align}
In fact, both sides of equation~\eqref{eqn:FP-rel-quan} satisfy differential equation~\eqref{eqn:FP-formula} and they are equal at the initial condition $\mathbf{s}=0$.  
\section{Virasoro constraints for  Gromov-Witten invariants and Hodge integrals}\label{sec:Virasoro-gw-Hodge}
\subsection{Virasoro constraints for descendant Gromov-Witten invariants}\label{subsec:virasoro-gw-inva}
In this subsection, we recall the constructions of Virasoro operators by Eguchi, Hori, and
Xiong, modified by Katz.
For our purpose, we define 
\[b_\alpha:=p_\alpha-\frac{d-1}{2},\,\, b^\alpha:=1-b_\alpha\]
and functions
\[\hat{e}_j(n,\alpha)(x)=e_j(x+b_\alpha,...,x+n+b_\alpha),\quad \check{e}_j(n,\alpha)(x)=e_j(x+b^\alpha,...,x+n+b^\alpha).\]
Here $e_j(x_0,...,x_n)=[t^{n+1-j}]\prod_{i=0}^{n}(t+x_i)$ denotes the $j$-th elementary symmetric function of $x_0,...,x_n$ and $[t^{n+1-j}]$ denotes taking the coefficient of $t^{n+1-j}$. Let $\mathcal{C}=(\mathcal{C}_\alpha^\beta)$ be the matrix of
multiplication by the first Chern class $c_1(X)$ in the ordinary cohomology ring, i.e.
$c_1(X)\cup \phi_\alpha=\sum_{\beta}\mathcal{C}_\alpha^\beta\phi_\beta$.
The differential operators are defined to be
\begin{align*}
L_n
=&\sum_{j=0}^{n+1}\sum_{r,\alpha}\left(\hat{e}_{n+1-j}(n,\alpha)\right)(r) \tilde{t}_r^\alpha \tau_{r+n-j}(c_1(X)^{j}\cup\phi_\alpha)
\\&-\frac{\hbar^2}{2}\sum_{j=0}^{n+1}\sum_{s,\alpha}(-1)^{-s}\left(\check{e}_{n+1-j}(n,\alpha)\right)(-s-1)\tau_s(\phi_\alpha)\tau_{-s-1+n-j}(c_1(X)^{j}\cup\phi^{\alpha})
\\&+\frac{1}{2\hbar^2}\sum_{\alpha,\beta}(\mathcal{C}^{n+1})_{\alpha\beta}t_0^\alpha t_0^\beta
+\frac{\delta_{n,0}}{24}
\int_{X}\left(\frac{3-d}{2}c_d(X)-c_1(X)\cup c_{d-1}(X)\right),\quad n\geq-1
\end{align*}
which satisfy the Virasoro bracket  relation
\begin{align*}
[L_n,L_m]=(n-m)L_{m+n},\quad n\geq-1.    
\end{align*}
The Virasoro conjecture states that: 
For any smooth projective variety, 
\[L_n\mathcal{D}(\mathbf{t})=0, \quad n\geq-1.\]
It is well known that for any compact symplectic manifold $L_n\mathcal{D}(\mathbf{t})=0$ for $n=-1$ or $0$.  The first
equation (i.e. for $n=-1$) is the string equation. The second equation (i.e. for $n=0$)
is derived from  divisor equation, dilaton equation and selection rule.

Since the appearence of Virasoro conjecture, a lot of  progress has been made. In \cite{MR1690740}, it was first proved that genus-0 Virasoro constraints hold for any compact symplectic manifold (see also \cite{MR1740678}, \cite{MR1718143}, \cite{MR2115767}). The genus-1 Virasoro constraints for manifolds with semisimple quantum cohomology was proved by Dubrovin and Zhang \cite{MR1740678} (see also \cite{MR1815723} and \cite{MR2234884}). The genus-2 Virasoro constraints for manifolds with semisimple quantum cohomology was proved by Liu by solving universal equations (cf. \cite{MR2306043}). In \cite{MR1866444}, using localization techniques, Givental proposed a quantization formula for descendant potential function and  proved all genus Virasoro Conjecture for toric Fano manifolds. Afterthat, using techniques in topological field theory, Teleman proved Givental's formula and Virasoro constraints hold for  all compact K{\"a}hler manifolds with semisimple quantum cohomology algebras in \cite{MR2917177}.  It has also been proved in \cite{MR2208418} all genus Virasoro conjecture hold  
for any nonsingular curves, using degeneration techniques and  relative Gromov-Witten theory.
% \subsection{Virasoro conjecture for Hodge integrals}
\subsection{Virasoro operators for Hodge integrals}
Formally, we define differential operators
\begin{align*}
&L_n^{\mathbb{E}}:=\exp(\sum_{k\geq1}\frac{B_{2k}}{(2k)!}\mathfrak{s}_{k}\widehat{z^{2k-1}})L_n\exp(-\sum_{k\geq1}\frac{B_{2k}}{(2k)!}\mathfrak{s}_{k}\widehat{z^{2k-1}}).
\end{align*}
It is obvious that operators $\{L_n^{\mathbb{E}}\}_{n\geq-1}$ satisfies Virasoro bracket relation
\begin{align*}
[L_n^{\mathbb{E}},L_m^{\mathbb{E}}]=(n-m)L_{n+m}^{\mathbb{E}} .   
\end{align*}
Then the  Virasoro conjecture for pure descendant Gromov-Witten invariants
\begin{align*}
L_n \mathcal{D}(\mathbf{t})=0    
\end{align*}
is equivalent to the Virasoro conjecture for Hodge integral
\begin{align*}
L_n^{\mathbb{E}}\mathcal{D}^{\mathbb{E}}(\mathbf{t},\mathbf{s})=0.    
\end{align*}
The explicit formula of operator $L_n^{\mathbb{E}}$ is given in the following proposition.
\begin{proposition} For $n\geq-1$,
\begin{align}\label{eqn:expression-L_n^E}
L_n^{\mathbb{E}}
=&L_n+\delta_{n,-1}\frac{\chi(X)}{24}\mathfrak{s}_{1}+\sum_{m=1}^{n+1}\frac{(-1)^{m}}{m!}\sum_{k_1,...,k_m\ge1}\left(\prod_{i=1}^{m}\left(\frac{B_{2k_i}}{(2k_i)!}\mathfrak{s}_{k_i}\right)\right)\nonumber
\\&\hspace{10pt}\cdot\Bigg(\sum_{j=0}^{n+1}\sum_{r,\alpha}\left(\prod_{i=1}^{m}\Delta_{2k_i-1}\hat{e}_{n+1-j}(n,\alpha)\right)(r) \tilde{t}_r^\alpha \tau_{\sum_{i=1}^{m}(2k_i-1)+r+n-j}(c_1(X)^{j}\cup\phi_\alpha)\nonumber
\\&\hspace{10pt}-\frac{\hbar^2}{2}\sum_{j=0}^{n+1}\sum_{s,\beta}(-1)^{-s}\left(\prod_{i=1}^{m}\Delta_{2k_i-1}\check{e}_{n+1-j}(n,\beta)\right)(-s-1)\tau_s(\phi_\beta)\tau_{\sum_{i=1}^{m}(2k_i-1)-s-1+n-j}(c_1(X)^{j}\cup\phi^{\beta})\Bigg).
\end{align}
Here $\chi(X)$ is the Euler characteristic  of  $X$. In particular,
\begin{align*}
L_0^{\mathbb{E}}
=&L_0-\sum_{k\ge1}\frac{(2k-1)B_{2k}}{(2k)!}\mathfrak{s}_{k}\Bigg(\sum_{r,\alpha} \tilde{t}_r^\alpha \tau_{r+2k-1}(\phi_\alpha)
-\frac{\hbar^2}{2}\sum_{s,\beta}(-1)^{-s}\tau_s(\phi_\beta)\tau_{2k-2-s}(\phi^{\beta})\Bigg).
\end{align*}
\end{proposition}
\begin{proof}
We firstly focus on the cases $n\geq0$. 
By the Baker-Campbell-Hausdorff formula, 
\begin{align*}
\exp(A)B\exp(-A)=\exp(ad_{A})(B) 
\end{align*}
we have
\begin{align}\label{eqn:L_nE=L_n+L_nk1km}
&L_n^{\mathbb{E}}=\exp(\sum_{k\geq1}\frac{B_{2k}}{(2k)!}\mathfrak{s}_{k}\widehat{z^{2k-1}})L_n\exp(-\sum_{k\geq1}\frac{B_{2k}}{(2k)!}\mathfrak{s}_{k}\widehat{z^{2k-1}})\nonumber
% \\=&L_n+[\sum_{k\geq1}\frac{B_{2k}}{(2k)!}\mathfrak{s}_k\widehat{z^{2k-1}},L_n]+\frac{1}{2!}[\sum_{k\geq1}\frac{B_{2k}}{(2k)!}\mathfrak{s}_k\widehat{z^{2k-1}},[\sum_{k\geq1}\frac{B_{2k}}{(2k)!}\mathfrak{s}_k\widehat{z^{2k-1}},L_n]]+...
\\=&L_n+\sum_{m=1}^{\infty}\frac{1}{m!}\sum_{k_1,...,k_m\geq1}\frac{B_{2k_1}}{(2k_1)!}...\frac{B_{2k_m}}{(2k_m)!}\mathfrak{s}_{k_1}...\mathfrak{s}_{k_m}[\widehat{z^{2k_1-1}},[\widehat{z^{2k_2-1}},...,[\widehat{z^{2k_m-1}},L_n]...]]\nonumber
\\=&L_n+\sum_{m=1}^{\infty}\frac{(-1)^{m}}{m!}\sum_{k_1,...,k_m\ge1}\left(\prod_{i=1}^{m}\left(\frac{B_{2k_i}}{(2k_i)!}\mathfrak{s}_{k_i}\right)\right)L_{n;k_1,...,k_m}
\end{align}  
where  operators $L_{n;k_1,...,k_m}$ is defined by
\[L_{n;k_1,...,k_m}:=[...[[L_n, \widehat{z^{2k_1-1}}],\widehat{z^{2k_2-1}}],...,\widehat{z^{2k_m-1}}].\]
We will see below the summation over $m$ in equation~\eqref{eqn:L_nE=L_n+L_nk1km} is finite. 
Define operators 
\begin{align*}
l_{n;k_1,...,k_m}:= \{...\{\{ l_n,z^{2k_1-1}\},z^{2k_2-1}\},...,z^{2k_m-1}\}   
\end{align*}
whose quantization is $-L_{n;k_1,...,k_m}$. 
It is direct to compute 
\begin{align*}
& l_{n;k_1,...,k_m}(z^r\phi_\alpha) \\=&\sum_{j=0}^{n+1}\left(\Delta_{2k_1-1}...\Delta_{2k_m-1}\hat{e}_{n+1-j}(n,\alpha)\right)(r) z^{\sum_{i=1}^{m}(2k_i-1)+r+n-j} c_1(X)^j\cup\phi_\alpha
\end{align*}
For a generic point  $f=\sum_{r,\alpha}q_r^\alpha z^r\phi_\alpha+\sum_{s,\beta}p_s^\beta(-z)^{-s-1}\phi^\beta$ in $\mathbb{H}((z))$, we have
\begin{align*}
&l_{n;k_1,...,k_m}f
\\=&\sum_{j=0}^{n+1}\sum_{r,\alpha}\left(\Delta_{2k_1-1}...\Delta_{2k_m-1}\hat{e}_{n+1-j}(n,\alpha)\right)(r) q_r^\alpha z^{\sum_{i=1}^{m}(2k_i-1)+r+n-j}c_1(X)^j\cup\phi_\alpha
\\&+\sum_{j=0}^{n+1}\sum_{s,\beta}(-1)^{-s-1}\left(\Delta_{2k_1-1}...\Delta_{2k_m-1}\check{e}_{n+1-j}(n,\beta)\right)(-s-1)p_s^{\beta}z^{\sum_{i=1}^{m}(2k_i-1)-s-1+n-j}c_1(X)^j\cup\phi^\beta.
\end{align*}
So the quadratic hamitonian of $l_{n;k_1,...,k_m}$ equals to 
\begin{align*}
&h_{l_{n;k_1,...,k_m}}(f)=\frac{1}{2}\Omega(l_{n;k_1,...,k_m}f|_{z\rightarrow-z},f)
\\=&-\sum_{j=0}^{n+1}\sum_{r,\alpha,\beta}\left(\Delta_{2k_1-1}...\Delta_{2k_m-1}\hat{e}_{n+1-j}(n,\alpha)\right)(r)q_r^\alpha p_{\sum_{i=1}^{m}(2k_i-1)+r+n-j}^\beta \langle c_1(X)^j\cup\phi_\alpha,\phi^\beta\rangle 
\\&+\frac{1}{2}\sum_{j=0}^{n+1}\sum_{s,\beta,\beta'}(-1)^{-s}\left(\Delta_{2k_1-1}...\Delta_{2k_m-1}\check{e}_{n+1-j}(n,\beta)\right)(-s-1)p_s^{\beta}p_{\sum_{i=1}^{m}(2k_i-1)-s-1+n-j}^{\beta'}\langle c_1(X)^{j}\cup\phi^\beta,\phi^{\beta'}\rangle
\\&-\frac{1}{2}\delta_{m,0}\sum_{\alpha,\beta} q_0^\alpha q_0^\beta\langle c_1(X)^{n+1}\cup\phi_\alpha,\phi_\beta\rangle.
 \end{align*}
Its  quantization equals to 
\begin{align*}
&-L_{n;k_1,...,k_m}
\\=&-\sum_{j=0}^{n+1}\sum_{r,\alpha,\beta}\left(\prod_{i=1}^{m}\Delta_{2k_i-1}\hat{e}_{n+1-j}(n,\alpha)\right)(r) \langle c_1(X)^{j}\cup\phi_\alpha,\phi^\beta\rangle\tilde{t}_r^\alpha \frac{\partial}{\partial t^\beta_{\sum_{i=1}^{m}(2k_i-1)+r+n-j}}
\\&+\frac{\hbar^2}{2}\sum_{j=0}^{n+1}\sum_{s,\beta,\beta'}(-1)^{-s}\left(\prod_{i=1}^{m}\Delta_{2k_i-1}\check{e}_{n+1-j}(n,\beta)\right)(-s-1)\langle c_1(X)^{j}\cup\phi^{\beta},\phi^{\beta'}\rangle
\frac{\partial}{\partial t_s^\beta}\frac{\partial}{\partial t_{\sum_{i=1}^{m}(2k_i-1)-s-1+n-j}^{\beta'}}
\\&-\frac{1}{2\hbar^2}\delta_{m,0}\sum_{\alpha,\beta} t_0^\alpha t_0^\beta\langle c_1(X)^{n+1}\cup\phi_\alpha,\phi_\beta\rangle.
\end{align*}
Notice for $m=0$, the operator $L_{n;\emptyset}$ is slightly different from Virasoro operator $L_n$ by a constant 
\begin{align*}
L_n=L_{n;\emptyset}+\frac{\delta_{n,0}}{24}\int_{X}\left(\frac{3-d}{2}c_d(X)-c_1(X)\cup c_{d-1}(X)\right).    
\end{align*}

For $m\geq1$, under the identification $\tau_n(\phi_\alpha)=\frac{\partial}{\partial t_n^\alpha}$,
\begin{align*}
&L_{n;k_1,...,k_m}
\\=&\sum_{j=0}^{n+1}\sum_{r,\alpha}\left(\prod_{i=1}^{m}\Delta_{2k_i-1}\hat{e}_{n+1-j}(n,\alpha)\right)(r) \tilde{t}_r^\alpha \tau_{\sum_{i=1}^{m}(2k_i-1)+r+n-j}(c_1(X)^{j}\cup\phi_\alpha)
\\&-\frac{\hbar^2}{2}\sum_{j=0}^{n+1}\sum_{s,\beta}(-1)^{-s}\left(\prod_{i=1}^{m}\Delta_{2k_i-1}\check{e}_{n+1-j}(n,\beta)\right)(-s-1)\tau_s(\phi_\beta)\tau_{\sum_{i=1}^{m}(2k_i-1)-s-1+n-j}(c_1(X)^{j}\cup\phi^{\beta}).
\end{align*}
So we have for $n\geq0$
\begin{align*}
L_n^{\mathbb{E}}
=&L_n+\sum_{m=1}^{n+1}\frac{(-1)^{m}}{m!}\sum_{k_1,...,k_m\ge1}\left(\prod_{i=1}^{m}\left(\frac{B_{2k_i}}{(2k_i)!}\mathfrak{s}_{k_i}\right)\right)
\\&\hspace{10pt}\cdot\Bigg(\sum_{j=0}^{n+1}\sum_{r,\alpha}\left(\prod_{i=1}^{m}\Delta_{2k_i-1}\hat{e}_{n+1-j}(n,\alpha)\right)(r) \tilde{t}_r^\alpha \tau_{\sum_{i=1}^{m}(2k_i-1)+r+n-j}(c_1(X)^{j}\cup\phi_\alpha)
\\&\hspace{10pt}-\frac{\hbar^2}{2}\sum_{j=0}^{n+1}\sum_{s,\beta}(-1)^{-s}\left(\prod_{i=1}^{m}\Delta_{2k_i-1}\check{e}_{n+1-j}(n,\beta)\right)(-s-1)\tau_s(\phi_\beta)\tau_{\sum_{i=1}^{m}(2k_i-1)-s-1+n-j}(c_1(X)^{j}\cup\phi^{\beta})\Bigg).
\end{align*}
For $n=-1$, things become subtle, since $l_{-1}=z^{-1}$ which is a negative power of $z$. For a generic point $f=\sum_{r,\alpha}q_r^\alpha z^r\phi_\alpha+\sum_{s,\beta}p_s^\beta(-z)^{-s-1}\phi^\beta$ in $\mathbb{H}((z))$, we have
\begin{align*}
&h_{l_{-1;\emptyset}}(f)
=-\sum_{r,\alpha}q_r^\alpha p_{r-1}^\alpha 
-\frac{1}{2}\sum_{\alpha,\beta}\eta_{\alpha\beta} q_0^\alpha q_0^\beta,
 \end{align*}
and
\begin{align*}
&h_{z^{2k_1-1}}(f)=-\sum_{r,\alpha}q_r^\alpha p_{r+2k_1-1}^\alpha 
+\frac{1}{2}\sum_{i,\alpha,\beta}(-1)^{i}\eta^{\alpha\beta} p_i^\alpha p_{2k_1-2-i}^\beta. 
\end{align*}
% It is easy to compute
% \begin{align*}
% C(h_{l_{-1;\emptyset}},h_{z^{2k_1-1}})=\frac{1}{2}str(Id)    
% \end{align*}
So
\begin{align*}
&C(h_{l_{-1;\emptyset}},h_{z^{2k_1-1}}) 
=-\frac{1}{4}\sum_{i,\alpha,\beta,\alpha',\beta'}\eta_{\alpha\beta}(-1)^{i}\eta^{\alpha'\beta'}C(q_0^\alpha q_0^{\beta}, p_i^{\alpha'} p^{\beta'}_{2k_1-2-i})
=\frac{1}{2}\delta_{k_1}^{1}\chi(X).\footnotemark
\end{align*}
\footnotetext{There is a subtlety here: we always assume $H^{odd}(X=0)$ for simplicity. For general case, all elements in Givental quantization are $\mathbb{Z}_2$ graded (cf. \cite{MR2276766}).}
Therefore,
\begin{align*}
L_{-1}^{\mathbb{E}}
=&L_{-1}+\frac{\chi(X)}{24}\mathfrak{s}_{1}.   
\end{align*}
\end{proof}
\begin{remark}
The explicit formula of $L_n^{\mathbb{E}}$ is an infinite sum over indices $k_1,...,k_{n+1}$ and $r$. It can be viewed as a formal power series of parameters $\mathbf{s}=(\mathfrak{s}_1,\mathfrak{s}_2,...)$.
\end{remark}
\begin{remark}
Formally, the  Virasoro conjecture for Hodge integrals  is equivalent to the Virasoro conjecture for descendant Gromov-Witten invariants. 
Despite these Virasoro conjectures are equivalent, the former one contains much more information, such as Hodge intgeral. 
It provides a large family of differential equations among Hodge integrals and descendant Gromov-Witten invariants for target varieties with semisimple quantum cohomology. Since there are infinite many parameters $\{\mathfrak{s}_k\}$, it is interesting to investigate to what  extent the Virasoro constraints for Hodge integrals hold for any smooth projective variety?  
\end{remark}
\begin{remark} 
A different expression for the  operator $L_n^{\mathbb{E}}$ was obtained  in \cite{MR4631415} for any calibrated semisimple Frobenius manifold. In fact, the authors in \cite{MR4631415} defined a Virasoro-like infinite dimensional Lie algebra and expressed  $L_n^{\mathbb{E}}$  as a linear combination of its basis. As a result, we see
Theorem~\ref{thm:hodge-semisimple} with $L_n^{\mathbb{E}}$ having a different expression  is a special case of \cite[Theorem 3]{MR4631415}. For the case when $X$ is a point, Theorem~\ref{thm:hodge-semisimple} was obtained in
\cite{zhouHodge}.
% In particular, if target $X$ equals to a point,  formula~\eqref{eqn:expression-L_n^E} recovers the similar formula considered in \cite{zhouHodge}. 
\end{remark}
Let $\Psi_{g,n}^{\mathbb{E}}(\mathbf{t},\mathbf{s})$ be the coefficient of $\hbar^{2g-2}$
in the Laurent expansion of $\mathcal{D}^{\mathbb{E}}(\mathbf{t},\mathbf{s})^{-1}L_n^{\mathbb{E}}\mathcal{D}^{\mathbb{E}}(\mathbf{t},\mathbf{s})$. That is
\[\mathcal{D}^{\mathbb{E}}(\mathbf{t},\mathbf{s})^{-1}L_n^{\mathbb{E}}\mathcal{D}^{\mathbb{E}}(\mathbf{t},\mathbf{s})
=\sum_{g\geq0}\hbar^{2g-2}\Psi_{g,n}^{\mathbb{E}}(\mathbf{t},\mathbf{s}).\]
Moreover, we can expand $\Psi_{g,n}^{\mathbb{E}}(\mathbf{t},\mathbf{s})$ as a formal power series of $\mathbf{s}$
\[\Psi_{g,n}^{\mathbb{E}}(\mathbf{t},\mathbf{s})
=\sum_{m\geq0}\sum_{k_1,...,k_m\geq0}
\frac{1}{m!}\mathfrak{s}_{k_1}...\mathfrak{s}_{k_m}\Psi_{g,n;k_1,...k_m}^{\mathbb{E}}(\mathbf{t}).\]
Therefore, the vanishing of $\Psi_{g,n}^{\mathbb{E}}(\mathbf{t},\mathbf{s})$ is equivalent to the vanishing of $\Psi_{g,n;k_1,...,k_m}^{\mathbb{E}}(\mathbf{t})$ for all $k_1,...,k_m$. Furthermore, the vanishing of $\Psi_{g,n;\emptyset}^{\mathbb{E}}(\mathbf{t})$ is exactly the Virasoro constraints for descendant Gromov-Witten invariants. 
\begin{proposition}
    \label{pro:psi-gnE}For $n\geq1$, 
\begin{align}\label{eqn:psignEts-formula}
&\Psi_{g,n}^{\mathbb{E}}(\mathbf{t},\mathbf{s})\nonumber   \\=&\sum_{m=0}^{n+1}\frac{(-1)^{m}}{m!}\sum_{k_1,...,k_m\ge1}\left(\prod_{i=1}^{m}\left(\frac{B_{2k_i}}{(2k_i)!}\mathfrak{s}_{k_i}\right)\right)\left\{\sum_{j=0}^{n+1}\sum_{r,\alpha,\beta}\left(\prod_{i=1}^{m}\Delta_{2k_i-1}\hat{e}_{n+1-j}(\alpha)\right)(r)\langle c_1(X)^{j}\cup\phi_\alpha,\phi^\beta\rangle\nonumber
\right.\\&\left.\hspace{20pt}\cdot\tilde{t}_r^\alpha \sum_{l\geq0}\sum_{a_1,...,a_l\geq1}\frac{\mathfrak{s}_{a_1}...\mathfrak{s}_{a_l}}{l!}\frac{\partial F_{g;a_1,...,a_l}^{\mathbb{E}}(\mathbf{t})}{\partial t^\beta_{\sum_{i=1}^{m}(2k_i-1)+r+n-j}}\nonumber
\right.\\&\left.-\frac{1}{2}\sum_{g_1+g_2=g}\sum_{j=0}^{n+1}\sum_{s,\beta,\beta'}(-1)^{-s}\left(\prod_{i=1}^{m}\Delta_{2k_i-1}\check{e}_{n+1-j}(\beta)\right)(-s-1)\langle c_1(X)^{j}\cup\phi^{\beta},\phi^{\beta'}\rangle\nonumber
\right.\\&\left.\hspace{50pt}\cdot
\left(\sum_{l\geq0}\sum_{a'_1,...,a'_l\geq1}\frac{\mathfrak{s}_{a'_1}...\mathfrak{s}_{a'_l}}{l!}\frac{\partial F_{g_1;a'_1,...,a'_l}^{\mathbb{E}}(\mathbf{t})}{\partial t_s^\beta}\right)\left(\sum_{l\geq0}\sum_{a''_1,...,a''_l\geq1}\frac{\mathfrak{s}_{a''_1}...\mathfrak{s}_{a''_l}}{l!}\frac{\partial F_{g_2;a''_1,...,a''_l}^{\mathbb{E}}(\mathbf{t})}{\partial t_{\sum_{i=1}^{m}(2k_i-1)-s-1+n-j}^{\beta'}} 
\right)\nonumber
\right.\\&\left.-\frac{1}{2}\sum_{j=0}^{n+1}\sum_{s,\beta,\beta'}(-1)^{-s}\left(\prod_{i=1}^{m}\Delta_{2k_i-1}\check{e}_{n+1-j}(\beta)\right)(-s-1)\langle c_1(X)^{j}\cup\phi^{\beta},\phi^{\beta'}\rangle
\nonumber\right.\\&\left.
\hspace{50pt}\cdot\sum_{l\geq0}\sum_{a_1,...,a_l\geq1}\frac{\mathfrak{s}_{a_1}...\mathfrak{s}_{a_l}}{l!}\frac{\partial^2 F_{g-1;a_1,...,a_l}^{\mathbb{E}}(\mathbf{t})}{\partial t_s^\beta\partial t_{\sum_{i=1}^{m}(2k_i-1)-s-1+n-j}^{\beta'}}\right\}
\nonumber\\&+\frac{1}{2}\delta_{g,0}\sum_{\alpha,\beta} t_0^\alpha t_0^\beta\langle c_1(X)^{n+1}\cup\phi_\alpha,\phi_\beta\rangle,
\end{align} 
where $F_{g;k_1,...,k_m}^{\mathbb{E}}$ is the generating function defined as follows
\[F_{g;k_1,...,k_m}^{\mathbb{E}}:=\langle\exp(\sum_{n,\alpha}t_n^\alpha\tau_n(\phi_\alpha));\ch_{2k_1-1}(\mathbb{E})...\ch_{2k_m-1}(\mathbb{E})\rangle_g.\]
\end{proposition}
\begin{proof}
By definition
\begin{align*}
&F_g^{\mathbb{E}}(\mathbf{t},\mathbf{s})   =\langle\exp(\sum_{n,\alpha}t_n^\alpha\tau_n(\phi_\alpha));\exp(\sum_{k\geq1}\mathfrak{s}_{k}\ch_{2k-1}(\mathbb{E}))\rangle_g
\\=&\sum_{m\geq0}\frac{1}{m!}\sum_{k_1,...,k_m\geq1}\mathfrak{s}_{k_1}...\mathfrak{s}_{k_m}\langle\exp(\sum_{n,\alpha}t_n^\alpha\tau_n(\phi_\alpha));\ch_{2k_1-1}(\mathbb{E})...\ch_{2k_m-1}(\mathbb{E})\rangle_g
\\=&\sum_{m\geq0}\frac{1}{m!}\sum_{k_1,...,k_m\geq1}\mathfrak{s}_{k_1}...\mathfrak{s}_{k_m}F_{g;k_1,...,k_m}^{\mathbb{E}}(\mathbf{t}),
\end{align*}
then
\begin{align*}
&\mathcal{D}^{\mathbb{E}}(\mathbf{t},\mathbf{s})=\exp(\sum_{g\geq0}\hbar^{2g-2}F_g^{\mathbb{E}}(\mathbf{t},\mathbf{s})) =\exp\left(\sum_{m\geq0}\frac{1}{m!}\sum_{k_1,...,k_m\geq1}\mathfrak{s}_{k_1}...\mathfrak{s}_{k_m}\sum_{g\geq0}\hbar^{2g-2}F_{g;k_1,...,k_m}^{\mathbb{E}}(\mathbf{t})\right).   
\end{align*}
Taking derivatives of $\mathcal{D}^{\mathbb{E}}(\mathbf{t},\mathbf{s})$, we get 
\begin{align*}
\mathcal{D}^{\mathbb{E}}(\mathbf{t},\mathbf{s})^{-1}\frac{\partial}{\partial t_n^\alpha}\mathcal{D}^{\mathbb{E}}(\mathbf{t},\mathbf{s})  =\sum_{m\geq0}\frac{1}{m!}\sum_{k_1,...,k_m\geq1}\mathfrak{s}_{k_1}...\mathfrak{s}_{k_m}\sum_{g\geq0}\hbar^{2g-2}\frac{\partial}{\partial t_n^\alpha}F_{g;k_1,...,k_m}^{\mathbb{E}}(\mathbf{t})  
\end{align*}
and
\begin{align*}
&\mathcal{D}^{\mathbb{E}}(\mathbf{t},\mathbf{s})^{-1}\frac{\hbar^2}{2}\frac{\partial}{\partial t_n^\alpha}\frac{\partial}{\partial t_l^\beta}\mathcal{D}^{\mathbb{E}}(\mathbf{t},\mathbf{s})  \\=&\sum_{m\geq0}\frac{1}{m!}\sum_{k_1,...,k_m\geq1}\mathfrak{s}_{k_1}...\mathfrak{s}_{k_m}\sum_{g\geq0}\frac{\hbar^{2g}}{2}\frac{\partial}{\partial t_n^\alpha}\frac{\partial}{\partial t_l^\beta}F_{g;k_1,...,k_m}^{\mathbb{E}}(\mathbf{t})   
\\&+\frac{\hbar^2}{2}
\left(\sum_{m\geq0}\frac{1}{m!}\sum_{k_1,...,k_m\geq1}\mathfrak{s}_{k_1}...\mathfrak{s}_{k_m}\sum_{g\geq0}\hbar^{2g-2}\frac{\partial}{\partial t_n^\alpha}F_{g;k_1,...,k_m}^{\mathbb{E}}(\mathbf{t})\right)
\\&\hspace{20pt}\cdot\left(\sum_{m\geq0}\frac{1}{m!}\sum_{k_1,...,k_m\geq1}\mathfrak{s}_{k_1}...\mathfrak{s}_{k_m}\sum_{g\geq0}\hbar^{2g-2}\frac{\partial}{\partial t_l^\beta}F_{g;k_1,...,k_m}^{\mathbb{E}}(\mathbf{t})\right).
% \\&+\frac{\hbar^2}{2}
% \left(\sum_{m\geq0}\frac{1}{m!}\sum_{k_1,...,k_m\geq1}\mathfrak{s}_{k_1}...\mathfrak{s}_{k_m}\sum_{g\geq0}\hbar^{2g-2}\frac{\partial}{\partial t_n^\alpha}\frac{\partial}{\partial t_l^\beta}F_{g;k_1,...,k_m}^{\mathbb{E}}(\mathbf{t})\right)
\end{align*}
Recall that
 \begin{align*}
&L_{n;k_1,...,k_m}
\\=&\sum_{j=0}^{n+1}\sum_{r,\alpha,\beta}\left(\prod_{i=1}^{m}\Delta_{2k_i-1}\hat{e}_{n+1-j}(\alpha)\right)(r) \langle c_1(X)^{j}\cup\phi_\alpha,\phi^\beta\rangle\tilde{t}_r^\alpha \frac{\partial}{\partial t^\beta_{\sum_{i=1}^{m}(2k_i-1)+r+n-j}}
\\&-\frac{\hbar^2}{2}\sum_{j=0}^{n+1}\sum_{s,\beta,\beta'}(-1)^{-s}\left(\prod_{i=1}^{m}\Delta_{2k_i-1}\check{e}_{n+1-j}(\beta)\right)(-s-1)\langle c_1(X)^{j}\cup\phi^{\beta},\phi^{\beta'}\rangle
\frac{\partial}{\partial t_s^\beta}\frac{\partial}{\partial t_{\sum_{i=1}^{m}(2k_i-1)-s-1+n-j}^{\beta'}}
\\&+\frac{1}{2\hbar^2}\delta_{m,0}\sum_{\alpha,\beta} t_0^\alpha t_0^\beta\langle c_1(X)^{n+1}\cup\phi_\alpha,\phi_\beta\rangle,
\end{align*}
so we have
\begin{align*}
&\mathcal{D}^{\mathbb{E}}(\mathbf{t},\mathbf{s})^{-1}L_{n;k_1,...,k_m}\mathcal{D}^{\mathbb{E}}(\mathbf{t},\mathbf{s})   \\=&\sum_{l\geq0}\sum_{a_1,...,a_l\geq1}\frac{\mathfrak{s}_{a_1}...\mathfrak{s}_{a_l}}{l!}\sum_{g\geq0}\hbar^{2g-2}
\\&\cdot\sum_{j=0}^{n+1}\sum_{r,\alpha,\beta}\left(\prod_{i=1}^{m}\Delta_{2k_i-1}\hat{e}_{n+1-j}(\alpha)\right)(r) \langle c_1(X)^{j}\cup\phi_\alpha,\phi^\beta\rangle\tilde{t}_r^\alpha \frac{\partial F_{g;a_1,...,a_l}^{\mathbb{E}}(\mathbf{t})}{\partial t^\beta_{\sum_{i=1}^{m}(2k_i-1)+r+n-j}}
\\&-\sum_{g\geq0}\frac{\hbar^{2g-2}}{2}\sum_{g_1+g_2=g}\sum_{j=0}^{n+1}\sum_{s,\beta,\beta'}(-1)^{-s}\left(\prod_{i=1}^{m}\Delta_{2k_i-1}\check{e}_{n+1-j}(\beta)\right)(-s-1)\langle c_1(X)^{j}\cup\phi^{\beta},\phi^{\beta'}\rangle
\\&\cdot
\left(\sum_{l\geq0}\sum_{a'_1,...,a'_l\geq1}\frac{\mathfrak{s}_{a'_1}...\mathfrak{s}_{a'_l}}{l!}\frac{\partial F_{g_1;a'_1,...,a'_l}^{\mathbb{E}}(\mathbf{t})}{\partial t_s^\beta}\right)\left(\sum_{l\geq0}\sum_{a''_1,...,a''_l\geq1}\frac{\mathfrak{s}_{a''_1}...\mathfrak{s}_{a''_l}}{l!}\frac{\partial F_{g_2;a''_1,...,a''_l}^{\mathbb{E}}(\mathbf{t})}{\partial t_{\sum_{i=1}^{m}(2k_i-1)-s-1+n-j}^{\beta'}} 
\right)
\\&-\sum_{g\geq0}\frac{\hbar^{2g}}{2}\sum_{j=0}^{n+1}\sum_{s,\beta,\beta'}(-1)^{-s}\left(\prod_{i=1}^{m}\Delta_{2k_i-1}\check{e}_{n+1-j}(\beta)\right)(-s-1)\langle c_1(X)^{j}\cup\phi^{\beta},\phi^{\beta'}\rangle
\\&\cdot
\left(\sum_{l\geq0}\sum_{a_1,...,a_l\geq1}\frac{\mathfrak{s}_{a_1}...\mathfrak{s}_{a_l}}{l!}\frac{\partial^2 F_{g;a_1,...,a_l}^{\mathbb{E}}(\mathbf{t})}{\partial t_s^\beta\partial t_{\sum_{i=1}^{m}(2k_i-1)-s-1+n-j}^{\beta'}}
\right)
\\&+\frac{1}{2\hbar^2}\delta_{m,0}\sum_{\alpha,\beta} t_0^\alpha t_0^\beta\langle c_1(X)^{n+1}\cup\phi_\alpha,\phi_\beta\rangle.
\end{align*}
Taking the coefficient of $\hbar^{2g-2}$, we get
\begin{align*}
&[\hbar^{2g-2}]\mathcal{D}^{\mathbb{E}}(\mathbf{t},\mathbf{s})^{-1}L_{n;k_1,...,k_m}\mathcal{D}^{\mathbb{E}}(\mathbf{t},\mathbf{s}) \\=&\sum_{l\geq0}\sum_{a_1,...,a_l\geq1}\frac{\mathfrak{s}_{a_1}...\mathfrak{s}_{a_l}}{l!}
\sum_{j=0}^{n+1}\sum_{r,\alpha,\beta}\left(\prod_{i=1}^{m}\Delta_{2k_i-1}\hat{e}_{n+1-j}(\alpha)\right)(r) \langle c_1(X)^{j}\cup\phi_\alpha,\phi^\beta\rangle\tilde{t}_r^\alpha \frac{\partial F_{g;a_1,...,a_l}^{\mathbb{E}}(\mathbf{t})}{\partial t^\beta_{\sum_{i=1}^{m}(2k_i-1)+r+n-j}}
\\&-\frac{1}{2}\sum_{g_1+g_2=g}\sum_{j=0}^{n+1}\sum_{s,\beta,\beta'}(-1)^{-s}\left(\prod_{i=1}^{m}\Delta_{2k_i-1}\check{e}_{n+1-j}(\beta)\right)(-s-1)\langle c_1(X)^{j}\cup\phi^{\beta},\phi^{\beta'}\rangle
\\&\cdot
\left(\sum_{l\geq0}\sum_{a'_1,...,a'_l\geq1}\frac{\mathfrak{s}_{a'_1}...\mathfrak{s}_{a'_l}}{l!}\frac{\partial F_{g_1;a'_1,...,a'_l}^{\mathbb{E}}(\mathbf{t})}{\partial t_s^\beta}\right)\left(\sum_{l\geq0}\sum_{a''_1,...,a''_l\geq1}\frac{\mathfrak{s}_{a''_1}...\mathfrak{s}_{a''_l}}{l!}\frac{\partial F_{g_2;a''_1,...,a''_l}^{\mathbb{E}}(\mathbf{t})}{\partial t_{\sum_{i=1}^{m}(2k_i-1)-s-1+n-j}^{\beta'}} 
\right)
\\&-\frac{1}{2}\sum_{j=0}^{n+1}\sum_{s,\beta,\beta'}(-1)^{-s}\left(\prod_{i=1}^{m}\Delta_{2k_i-1}\check{e}_{n+1-j}(\beta)\right)(-s-1)\langle c_1(X)^{j}\cup\phi^{\beta},\phi^{\beta'}\rangle
\\&
\hspace{40pt}\cdot\sum_{l\geq0}\sum_{a_1,...,a_l\geq1}\frac{\mathfrak{s}_{a_1}...\mathfrak{s}_{a_l}}{l!}\frac{\partial^2 F_{g-1;a_1,...,a_l}^{\mathbb{E}}(\mathbf{t})}{\partial t_s^\beta\partial t_{\sum_{i=1}^{m}(2k_i-1)-s-1+n-j}^{\beta'}}
\\&+\frac{1}{2}\delta_{g,0}\delta_{m,0}\sum_{\alpha,\beta} t_0^\alpha t_0^\beta\langle c_1(X)^{n+1}\cup\phi_\alpha,\phi_\beta\rangle.
\end{align*} 
Combining with equation~\eqref{eqn:L_nE=L_n+L_nk1km}, we obtain equation~\eqref{eqn:psignEts-formula}.
\end{proof}

In this paper below, we will first prove the vanishing of $\Psi_{g,n;k_1,...,k_m}^{\mathbb{E}}(\mathbf{t})$ in genus $g=0$, then we  focus on proving the vanishing of $\Psi_{g,n;k_1}^{\mathbb{E}}(\mathbf{t})$ in genus $g\geq1$. From Proposition~\ref{pro:psi-gnE}, we see
\begin{corollary}\label{cor:formula-psi-gnk1E}
\begin{align}\label{eqn:formula-psi-gnk1E}
&\Psi^{\mathbb{E}}_{g,n;k_1}\nonumber   
\\=&-\frac{B_{2k_1}}{(2k_1)!}\left\{\sum_{j=0}^{n+1}\sum_{r,\alpha,\beta}\left(\Delta_{2k_1-1}\hat{e}_{n+1-j}(n,\alpha)\right)(r) \langle c_1(X)^{j}\cup\phi_\alpha,\phi^\beta\rangle\tilde{t}_r^\alpha \frac{\partial F_{g;\emptyset}^{\mathbb{E}}(\mathbf{t})}{\partial t^\beta_{(2k_1-1)+r+n-j}}\nonumber
\right.\\&\left.-\frac{1}{2}\sum_{g_1+g_2=g}\sum_{j=0}^{n+1}\sum_{s,\beta,\beta'}(-1)^{-s}\left(\Delta_{2k_1-1}\check{e}_{n+1-j}(n,\beta)\right)(-s-1)\langle c_1(X)^{j}\cup\phi^{\beta},\phi^{\beta'}\rangle
\frac{\partial F_{g_1;\emptyset}^{\mathbb{E}}(\mathbf{t})}{\partial t_s^\beta}\frac{\partial F_{g_2;\emptyset}^{\mathbb{E}}(\mathbf{t})}{\partial t_{(2k_1-1)-s-1+n-j}^{\beta'}} 
\nonumber
\right.\\&\left.-\frac{1}{2}\sum_{j=0}^{n+1}\sum_{s,\beta,\beta'}(-1)^{-s}\left(\Delta_{2k_1-1}\check{e}_{n+1-j}(n,\beta)\right)(-s-1)\langle c_1(X)^{j}\cup\phi^{\beta},\phi^{\beta'}\rangle
\frac{\partial^2 F_{g-1;\emptyset}^{\mathbb{E}}(\mathbf{t})}{\partial t_s^\beta\partial t_{(2k_1-1)-s-1+n-j}^{\beta'}}\right\}
\nonumber\\&+\left\{\sum_{j=0}^{n+1}\sum_{r,\alpha,\beta}\left(\hat{e}_{n+1-j}(n,\alpha)\right)(r) \langle c_1(X)^{j}\cup\phi_\alpha,\phi^\beta\rangle\tilde{t}_r^\alpha \frac{\partial F_{g;k_1}^{\mathbb{E}}(\mathbf{t})}{\partial t^\beta_{r+n-j}}
\nonumber\right.\\&\left.-\frac{1}{2}\sum_{g_1+g_2=g}\sum_{j=0}^{n+1}\sum_{s,\beta,\beta'}(-1)^{-s}\left(\check{e}_{n+1-j}(n,\beta)\right)(-s-1)\langle c_1(X)^{j}\cup\phi^{\beta},\phi^{\beta'}\rangle
\frac{\partial F_{g_1;k_1}^{\mathbb{E}}(\mathbf{t})}{\partial t_s^\beta}\frac{\partial F_{g_2;\emptyset}^{\mathbb{E}}(\mathbf{t})}{\partial t_{-s-1+n-j}^{\beta'}} 
\nonumber
\right.\\&\left.-\frac{1}{2}\sum_{g_1+g_2=g}\sum_{j=0}^{n+1}\sum_{s,\beta,\beta'}(-1)^{-s}\left(\check{e}_{n+1-j}(n,\beta)\right)(-s-1)\langle c_1(X)^{j}\cup\phi^{\beta},\phi^{\beta'}\rangle
\frac{\partial F_{g_1;\emptyset}^{\mathbb{E}}(\mathbf{t})}{\partial t_s^\beta}\frac{\partial F_{g_2;k_1}^{\mathbb{E}}(\mathbf{t})}{\partial t_{-s-1+n-j}^{\beta'}} 
\nonumber
\right.\\&\left.-\frac{1}{2}\sum_{j=0}^{n+1}\sum_{s,\beta,\beta'}(-1)^{-s}\left(\check{e}_{n+1-j}(n,\beta)\right)(-s-1)\langle c_1(X)^{j}\cup\phi^{\beta},\phi^{\beta'}\rangle
\frac{\partial^2 F_{g-1;k_1}^{\mathbb{E}}(\mathbf{t})}{\partial t_s^\beta\partial t_{-s-1+n-j}^{\beta'}}\right\}
\end{align}  
for $g\geq1$ and $n\geq1$. 
    
\end{corollary}

\section{Virasoro constraints for Hodge integral in  genus 0}\label{sec:g=0-Virasoro-Hodge}
Virasoro constraints for quantum cohomolgy of any compact symplectic manifold was firstly proved in \cite{MR1690740}.
Note that if the target is considered to be a symplectic manifold, usually
the holomorphic dimension $p_\alpha$ is replaced by a half of the real dimension of $\phi_\alpha$. In fact, genus zero Virasoro constraints are a consequence of the genus-0 WDVV equation and some basic differential equations including string equation, dilation equation, divisor equation and selection rule.  Besides the work of \cite{MR1690740}, there are some other proofs of genus-0 Virasoro constraints (see \cite{MR1740678}, \cite{MR1718143}, \cite{MR2115767}).  
\subsection{Proof of Theorem~\ref{thm:g=0-L_n^E-constraints}}
In this subsection, we give the proof of Theorem~\ref{thm:g=0-L_n^E-constraints}. By definition,
\begin{align*}
[\hbar^{-2}]\mathcal{D}^{\mathbb{E}}(\mathbf{t},\mathbf{s})^{-1}L_n^{\mathbb{E}}\mathcal{D}^{\mathbb{E}}(\mathbf{t},\mathbf{s})
=\Psi_{0,n}^{\mathbb{E}}(\mathbf{t},\mathbf{s}).    
\end{align*}
So Virasoro constraints for genus-0 Hodge integrals is equivalent to 
\begin{align*}
\Psi^{\mathbb{E}}_{0,n}(\mathbf{t},\mathbf{s}) =0,\quad n\geq-1. 
\end{align*}    
As pointed out in subsection~\ref{subsec:virasoro-gw-inva}, the vanishing of $\Psi^{\mathbb{E}}_{0,-1}(\mathbf{t},\mathbf{s})$ is the genus-0 string equation and the vanishing of $\Psi^{\mathbb{E}}_{0,0}(\mathbf{t},\mathbf{s})$ follows from genus-0  divisor equation and selection rule. In the below, we only focus on the cases $n\geq1$. 
Recall
\begin{align*}
&L_n^{\mathbb{E}}=L_n+\sum_{m=1}^{\infty}\frac{(-1)^{m}}{m!}\sum_{k_1,...,k_m\ge1}\left(\prod_{i=1}^{m}\left(\frac{B_{2k_i}}{(2k_i)!}\mathfrak{s}_{k_i}\right)\right)L_{n;k_1,...,k_m}
\end{align*}  
for $n\geq1$, 
where  the operator $L_{n;k_1,...,k_m}$ is defined by
\[L_{n;k_1,...,k_m}:=[...[[L_n, \widehat{z^{2k_1-1}}],\widehat{z^{2k_2-1}}],...,\widehat{z^{2k_m-1}}].\]
Formally, we define
\begin{align*}
\mathcal{D}(\mathbf{t})^{-1}L_{n;k_1,...,k_m}\mathcal{D}(\mathbf{t})=\sum_{g\geq0}\hbar^{2g-2}\Psi_{g,n;k_1,...,k_m}. \end{align*}
Then for $n\geq1$
\begin{align}\label{eqn:psiE0n=psi0nemp+}
\Psi^{\mathbb{E}}_{0,n}(\mathbf{t},\mathbf{s})=\Psi_{0,n;\emptyset}+\sum_{m=1}^{\infty}\frac{(-1)^{m}}{m!}\sum_{k_1,...,k_m\ge1}\left(\prod_{i=1}^{m}\left(\frac{B_{2k_i}}{(2k_i)!}\mathfrak{s}_{k_i}\right)\right)\Psi_{0,n;k_1,...,k_m}.
\end{align}
In the below, we will prove the vanishing of $\Psi_{0,n;k_1,...,k_m}$ by induction.  For $m=0$, the vanishing of $\Psi_{0,n;\emptyset}$ is  exactly  the Virasoro constraints in genus zero. Assume $\Psi_{0,n;k_1,...,k_m}=0$,
then
\begin{align*}
&[L_{n;k_1,...,k_m},\widehat{z^{2k_{m+1}-1}}]\mathcal{D}(\mathbf{t})
\\=&L_{n;k_1,...,k_m}\left(\mathcal{D}(\mathbf{t})\cdot\sum_{g\geq0}\hbar^{2g-2}\frac{(2k_{m+1})!}{B_{2k_{m+1}}}\langle\langle \ch_{2k_{m+1}-1}(\mathbb{E})\rangle\rangle_{g}\right)
-\widehat{z^{2k_{m+1}-1}}\left(\mathcal{D}(\mathbf{t})\cdot\sum_{g\geq0}\hbar^{2g-2}\Psi_{g,n;k_1,...,k_m}
\right).
\end{align*}
Multiplying $\exp(-\sum_{g\geq0}\hbar^{2g-2}F_g)$, we get
\begin{align*}
&\exp(-\sum_{g\geq0}\hbar^{2g-2}F_g)L_{n;k_1,...,k_{m+1}} \exp(\sum_{g\geq0}\hbar^{2g-2}F_g)
\\=&
\Big(-\exp(\sum_{g\geq0}\hbar^{2g-2}F_g)\Big)L_{n;k_1,...,k_m}\Big(\exp(\sum_{g\geq0}\hbar^{2g-2}F_g)\Big)\cdot \sum_{g\geq0}\hbar^{2g-2}\langle\langle \frac{(2k_{m+1})!}{B_{2k_{m+1}}}\ch_{2k_{m+1}-1}(\mathbb{E})\rangle\rangle_g
\\&+  \sum_{g\geq0}\hbar^{2g-2}\left(L_{n;k_1,...,k_m}-\delta_{m,0}\frac{1}{2\hbar^2}\sum_{\alpha,\beta}(\mathcal{C}^{n+1})_{\alpha\beta}t_0^\alpha t_0^\beta\right)\langle\langle \frac{(2k_{m+1})!}{B_{2k_{m+1}}}\ch_{2k_{m+1}-1}(\mathbb{E})\rangle\rangle_g
\\&+\sum_{g\geq0}\frac{\hbar^{2g-2}}{2}\sum_{g_1+g_2=g}\sum_{j=0}^{n}\sum_{s,\beta}(-1)^{-s}\left(\prod_{i=1}^{m}\Delta_{2k_i-1}\check{e}_{n+1-j}(\beta)\right)(-s-1)
\\&\hspace{120pt}\cdot\langle\langle\tau_s(\phi_\beta)\rangle\rangle_{g_1}\langle\langle\tau_{\sum_{i=1}^{m}(2k_i-1)-s-1+n-j}(c_1(X)^{j}\cup\phi^{\beta});\frac{(2k_{m+1})!}{B_{2k_{m+1}}}\ch_{2k_{m+1}-1}(\mathbb{E})\rangle\rangle_{g_2}
\\&+\sum_{g\geq0}\frac{\hbar^{2g-2}}{2}\sum_{g_1+g_2=g}\sum_{j=0}^{n}\sum_{s,\beta}(-1)^{-s}\left(\prod_{i=1}^{m}\Delta_{2k_i-1}\check{e}_{n+1-j}(\beta)\right)(-s-1)
\\&\hspace{120pt}\cdot\langle\langle\tau_s(\phi_\beta);\frac{(2k_{m+1})!}{B_{2k_{m+1}}}\ch_{2k_{m+1}-1}(\mathbb{E})\rangle\rangle_{g_1}\langle\langle\tau_{\sum_{i=1}^{m}(2k_i-1)-s-1+n-j}(c_1(X)^{j}\cup\phi^{\beta})\rangle\rangle_{g_2}
\\&-\left(\exp(-\sum_{g\geq0}\hbar^{2g-2}F_g)\right)\widehat{z^{2k_{m+1}-1}}\left(\exp(\sum_{g\geq0}\hbar^{2g-2}F_g)\right)
\cdot\sum_{g\geq0}\hbar^{2g-2}\Psi_{g,n;k_1,...,k_m}
\\&+\sum_{g\geq0}\sum_{i\geq0}\sum_{\alpha}\hbar^{2g-2}\tilde{t}_i^\alpha\frac{\partial\Psi_{g,n;k_1,...,k_m}}{\partial t_{i+2k_{m+1}-1}^\alpha}
-\sum_{g\geq 0}\sum_{i=0}^{2k_{m+1}-2}\sum_{\alpha,\beta}\frac{\hbar^{2g}}{2}(-1)^{i}\eta^{\alpha\beta}\frac{\partial^2\Psi_{g,n;k_1,...,k_m}}{\partial t_i^\alpha \partial t_{2k_{m+1}-2-i}^\beta}
\\&-\sum_{g\geq0}\sum_{i=0}^{2k_{m+1}-2}\hbar^{2g-2}
\sum_{g_1+g_2=g}\sum_{\alpha}(-1)^i
\langle\langle\tau_i(\phi^\alpha)\rangle\rangle_{g_1}\frac{\partial^2\Psi_{g_2,n;k_1,...,k_m}}{\partial t_{2k_{m+1}-2-i}^\alpha}.
\end{align*}
Choosing the coefficient of $\hbar^{-2}$
and by induction hypothesis, using the fact  $\ch_{2m-1}(\mathbb{E})=0$ for $m>0$ on the genus-0 moduli space of curves $\overline{\mathcal{M}}_{0,n}$, we prove
\begin{align*}
\Psi_{0,n;k_1,...,k_{m+1}}=0.    
\end{align*}
In the end, the vanishing of $\Psi_{0,n}^{\mathbb{E}}(\mathbf{t},\mathbf{s})$ follows from the vanishing of $\{\Psi_{0,n;k_1,...,k_{m}}:k_1,...,k_m\geq1\}_{m\geq0}$ and equation~\eqref{eqn:psiE0n=psi0nemp+}. 

\begin{remark}
In \cite{MR4631415}, certain Virasoro-like constraints were proved for the genus-0 potential function of any Frobenius manifold.  In fact, both proofs of Theorem~\ref{thm:g=0-L_n^E-constraints} and \cite[Theorem 2]{MR4631415} are based on the genus-0 Virasoro constraints. By \cite[Theorem 1]{MR4631415}, one can see the Virasoro constraints for genus-0 Hodge intergals is a certain linear combination of the Virasoro-like constraints in \cite{MR4631415}.
\end{remark}

\subsection{Eguchi-Hori-Xiong's $\widetilde{L}_{n+1}$ equation}
Virasoro constraints for Hodge integrals in genus-0 gives us a large family of differential equations of genus-0 Gromov-Witten invariants of any target variety 
\begin{align}\label{eqn:genus-0-Hodge-gw-equation}
&-\sum_{j=0}^{n}\sum_{r,\alpha}\left(\prod_{i=1}^{m}\Delta_{2k_i-1}\hat{e}_{n+1-j}(n,\alpha)\right)(r) \tilde{t}_r^\alpha \langle\langle\tau_{\sum_{i=1}^{m}(2k_i-1)+r+n-j}(c_1(X)^{j}\cup\phi_\alpha)
\rangle\rangle_0\nonumber
\\&+\frac{1}{2}\sum_{j=0}^{n}\sum_{s,\beta}(-1)^{-s}\left(\prod_{i=1}^{m}\Delta_{2k_i-1}\check{e}_{n+1-j}(n,\beta)\right)(-s-1)\langle\langle\tau_s(\phi_\beta)\rangle\rangle_0\langle\langle\tau_{\sum_{i=1}^{m}(2k_i-1)-s-1+n-j}(c_1(X)^{j}\cup\phi^{\beta})\rangle\rangle_0
\nonumber
\\&=0
\end{align}
for any $m,n,k_1,..., k_m\geq1$.

Recall that in \cite{MR1454328}, the following sequence of equations (for $n=0,1$) are conjectured for quantum cohomology of compact Fano manifolds

\begin{align}\label{eqn:genus-0-Lntilde}
&-\sum_{j=0}^{n}\sum_{r,\alpha}\left(\hat{e}_{n-j}(n-1,\alpha)\right)(r+1) \tilde{t}_r^\alpha \langle\langle\tau_{1+r+n-j}(c_1(X)^{j}\cup\phi_\alpha)
\rangle\rangle_0\nonumber
\\&+\frac{1}{2}\sum_{j=0}^{n}\sum_{s,\alpha}(-1)^{-s}\left(\check{e}_{n-j}(n-1,\alpha)\right)(-s)\langle\langle\tau_s(\phi_\alpha)\rangle\rangle_0\langle\langle\tau_{-s+n-j}(c_1(X)^{j}\cup\phi^{\alpha})\rangle\rangle_0=0,\hspace{30pt} n\geq0.
\end{align}
We refer to~\eqref{eqn:genus-0-Lntilde} as the Eguchi-Hori-Xiong's $\widetilde{L}_{n+1}$ equation

The  $\widetilde{L}_{n}$ equations were firstly proved in \cite{MR1690740} for $n=1,2$ and  then generalized to certain Virasoro-like constraints for genus-0 potential function of any Frobenius manifold in \cite{MR4631415}.   Below  we will see $\widetilde{L}_{n}$ equations are embedded in the genus-0 Virasoro constraints for Hodge integrals.
\begin{corollary}
For $n\geq1$, Eguchi-Hori-Xiong $\widetilde{L}_n$ constraint  holds for quantum cohomology of any target variety.  
\end{corollary}
\begin{proof}
Taking $m=1$ and $k_1=1$ in  equation~\eqref{eqn:genus-0-Hodge-gw-equation}, we obtain 
\begin{align}\label{eqn:genus-0-Hodge-gw-equation-Lntilde}
&-\sum_{j=0}^{n}\sum_{r,\alpha}\left(\Delta_{1}\hat{e}_{n+1-j}(n,\alpha)\right)(r) \tilde{t}_r^\alpha \langle\langle\tau_{1+r+n-j}(c_1(X)^{j}\cup\phi_\alpha)
\rangle\rangle_0\nonumber
\\&+\frac{1}{2}\sum_{j=0}^{n}\sum_{s,\alpha}(-1)^{-s}\left(\Delta_{1}\check{e}_{n+1-j}(n,\alpha)\right)(-s-1)\langle\langle\tau_s(\phi_\alpha)\rangle\rangle_0\langle\langle\tau_{-s+n-j}(c_1(X)^{j}\cup\phi^{\alpha})\rangle\rangle_0
\nonumber
\\=&0.
\end{align}
Up to a factor $n+1$, equation~\eqref{eqn:genus-0-Hodge-gw-equation-Lntilde} match exactly with equation~\eqref{eqn:genus-0-Lntilde}. 
In fact,
\begin{align*}
 &\left(\Delta_{1}\hat{e}_{n+1-j}(n,\alpha)\right)(r)  
\\=&e_{n+1-j}(r+1+b_\alpha,...,r+1+b_\alpha+n)-e_{n+1-j}(r+b_\alpha,...,r+b_\alpha+n)
\\=&(n+1)e_{n-j}(r+1+b_\alpha,...,r+b_\alpha+n)
\\=&(n+1)\hat{e}_{n-j}(n-1,\alpha)(r+1)
\end{align*}
and
\begin{align*}
 &\left(\Delta_{1}\check{e}_{n+1-j}(n,\alpha)\right)(-s-1)  
\\=&e_{n+1-j}(-s+b^\alpha,...,-s+b^\alpha+n)-e_{n+1-j}(-s-1+b^\alpha,...,-s-1+b^\alpha+n)
\\=&(n+1)e_{n-j}(-s+b^\alpha,...,-s-1+b^\alpha+n)
\\=&(n+1)\check{e}_{n-j}(n-1,\alpha)(-s).
\end{align*}
\end{proof}

% \begin{remark}
% If $m=1,k_1=1$, the vanishing of $\Psi_{0,n;k_1}$ is exactly the genus-0  $\widetilde{L}_n$ equations proposed by Eguchi,Hori and Xiong in \cite{MR1454328},     
% \end{remark}

\section{Proof of Theorem~\ref{thm:psiEg=1}}
\label{sec:proof-thm-psiEg=1}
In this section, we give the proof to Theorem~\ref{thm:psiEg=1}. For our purpose, we introduce a $T$ operator on the big phase space 
\begin{align*}
T(W)=\tau_+(W)-\sum_{\alpha}\langle\langle W\phi^\alpha\rangle\rangle_0\phi_\alpha
\end{align*}
for any vector field $W$ on the big phase space, where $\tau_+(W)$ is a linear operator defined
by $\tau_+(\tau_n(\phi_\alpha)):=\tau_{n+1}(\phi_\alpha)$. We will also use $\tau_k(W):=\tau_+^k(W)$. 
By induction, it is not difficult to obtain (cf. \cite[Equation (25)]{MR2775793})
\begin{align}\label{eqn:T^k(W)-tau_kW}
T^k(W)=\tau_k(W)-\sum_{i=0}^{k-1}\sum_{\alpha}
(-1)^{i}\langle\langle W\tau_i(\phi^\alpha)\rangle\rangle_0
\tau_{k-1-i}(\phi_\alpha).    
\end{align}
Another useful relation between operators $\tau_+$ and $T$
is:
for any contravariant tensors $P$ and $Q$ on the big phase space,
\begin{align}\label{eqn:tau_jtau_m-j=TjTm-j}
\sum_{j=0}^{m}\sum_{\alpha}(-1)^j P(\tau_j(\phi_\alpha)) Q(\tau_{m-j}(\phi^\alpha))
=\sum_{j=0}^{m}\sum_{\alpha}(-1)^j P(T^j(\phi_\alpha)) Q(T^{m-j}(\phi^\alpha))
\end{align}
for any $m\geq0$ (cf. \cite[Proposition 3.2]{MR2775793}).

\subsection{Quasi-homogenous equation of Euler vector field $\mathcal{X}$}
The third one is 
the Euler vector field on the big phase space  defined by
\[\mathcal{X}=-\sum_{n,\alpha}(n+b_\alpha-b_1-1)\tilde{t}_n^\alpha\tau_{n}(\phi_\alpha)-\sum_{n,\alpha}\tilde{t}_n^\alpha\tau_{n-1}(c_1(X)\cup\phi_\alpha).\]
Then the divisor equation for first Chern class $c_1(X)$ and selection rule gives the following quasi-homogenous equation for Euler vector field
\begin{align}\label{eqn:quasi-homo-eqn}
\langle\langle\mathcal{X}\rangle\rangle_g=   (3-d)(1-g)F_g+\frac{1}{2}\delta_{g,0}\sum_{\alpha,\beta}\mathcal{C}_{\alpha\beta}t_0^\alpha t_0^\beta -\frac{1}{24}\delta_{g,1}\int_{X}c_1(X)\cup c_{d-1}(X),
\end{align}
where $d$ is the complex dimension of $X$ and $c_i$
is the $i$-th Chern class.
Derivatives of quasi-homogenous equation~\eqref{eqn:quasi-homo-eqn} give us
\begin{lemma}\label{lem:der-quasi-homo-eqn}
\begin{align*}
&(i)\langle\langle\mathcal{X}\tau_{m}(\phi_\alpha)\rangle\rangle_0 =(m+b_\alpha+b_1+1)\langle\langle\tau_m(\phi_\alpha)
\rangle\rangle_0+\langle\langle\tau_{m-1}(c_1(X)\cup\phi_\alpha)\rangle\rangle_0+\delta_{m,0}\sum_{\beta}\mathcal{C}_{\alpha\beta}t_0^\beta,
\\&(ii)
\langle\langle\mathcal{X}\tau_m(\phi_\alpha)\tau_n(\phi_\beta)\rangle\rangle_0=(m+n+b_\alpha+b_\beta)\langle\langle\tau_m(\phi_\alpha)\tau_n(\phi_\beta)\rangle\rangle_0 +\delta_{m,0}\delta_{n,0}\mathcal{C}_{\alpha\beta}
\nonumber\\& \hspace{140pt}
+\langle\langle\tau_{m-1}(c_1(X)\cup\phi_\alpha)\tau_n(\phi_\beta)\rangle\rangle_0
+\langle\langle\tau_m(\phi_\alpha)\tau_{n-1}(c_1(X)\cup\phi_\beta)\rangle\rangle_0,
\\&(iii)\langle\langle\mathcal{X}\tau_{m}(\phi_\alpha)\tau_n(\phi_\beta)\tau_k(\phi_{\gamma})\rangle\rangle_0
=(m+n+k+b_\alpha+b_\beta+b_\gamma-\frac{3-d}{2})\langle\langle\tau_{m}(\phi_\alpha)\tau_n(\phi_\beta)\tau_k(\phi_\gamma)\rangle\rangle_0\nonumber
\\&\hspace{160pt}+
\langle\langle\tau_{m-1}(c_1(X)\cup\phi_\alpha)\tau_n(\phi_\beta)\tau_k(\phi_\gamma)\rangle\rangle_0 \nonumber
\\&\hspace{160pt}+\langle\langle\tau_m(\phi_\alpha)\tau_{n-1}(c_1(X)\cup\phi_\beta)\tau_k(\phi_\gamma)\rangle\rangle_0\nonumber
\\&\hspace{160pt}+\langle\langle\tau_{m}(\phi_\alpha)\tau_n(\phi_\beta)\tau_{k-1}(c_1(X)\cup \phi_\gamma)\rangle\rangle_0.
\end{align*}
for any $m,n,k$ and $\alpha,\beta,\gamma$.    
\end{lemma}
\subsection{Genus-0 Faber-Pandharipande equation}
Faber-Pandharipande equation~\eqref{eqn:FP-formula} in genus-0 is 
\begin{align}\label{eqn:g=0-FP-equation}
&-\sum_{r,\alpha}\Tilde{t}_r^\alpha\langle\langle\tau_{r+2k_1-1}(\phi_\alpha)\rangle\rangle_{0}+\frac{1}{2}\sum_{i=0}^{2k_1-2}\sum_{\alpha}(-1)^i\langle\langle\tau_i(\phi_\alpha)\rangle\rangle_0\langle\langle\tau_{2k_1-2-i}(\phi^\alpha)\rangle\rangle_0=0    
\end{align}
for $k_1>0$.

\begin{lemma}\label{lem:der-g=0-FP-eqn}For $k_1\geq1$, 
\begin{align*}
&(i) \sum_{r,\alpha}\Tilde{t}_r^\alpha\langle\langle\tau_{r+2k_1-1}(\phi_\alpha)\tau_{m}(\phi_\beta)\rangle\rangle_{0}
=\sum_{i,\alpha}(-1)^i\langle\langle\tau_i(\phi_\alpha)\tau_m(\phi_\beta)\rangle\rangle_0\langle\langle\tau_{2k_1-2-i}(\phi^\alpha)\rangle\rangle_0-\langle\langle\tau_{m+2k_1-1}(\phi_\beta)\rangle\rangle_0,
\\&(ii)\sum_{r,\alpha}\Tilde{t}_r^\alpha\langle\langle\tau_{r+2k_1-1}(\phi_\alpha)\tau_{m}(\phi_\beta)\tau_n(\phi_\sigma)\rangle\rangle_{0}
=\sum_{i,\alpha}(-1)^i\langle\langle\tau_i(\phi_\alpha)\tau_m(\phi_\beta)\tau_n(\phi_\sigma)\rangle\rangle_0\langle\langle\tau_{2k_1-2-i}(\phi^\alpha)\rangle\rangle_0,
\\&(iii)\sum_{r,\alpha}\Tilde{t}_r^\alpha\langle\langle\tau_{r+2k_1-1}(\phi_\alpha)\tau_{m}(\phi_\beta)\tau_n(\phi_\sigma)\tau_l(\phi_\epsilon)\rangle\rangle_{0}
=\sum_{i,\alpha}(-1)^i\langle\langle\tau_i(\phi_\alpha)\tau_m(\phi_\beta)\tau_n(\phi_\sigma)\tau_{l}(\phi_\epsilon)\rangle\rangle_0\langle\langle\tau_{2k_1-2-i}(\phi^\alpha)\rangle\rangle_0
\end{align*}
for any fixed $\beta,\sigma,\epsilon$.
\end{lemma}
\begin{proof}
Taking derivative of equation~\eqref{eqn:g=0-FP-equation} with respect to $t_m^\beta$, we get
\begin{align*}
&\sum_{r,\alpha}\Tilde{t}_r^\alpha\langle\langle\tau_{r+2k_1-1}(\phi_\alpha)\tau_{m}(\phi_\beta)\rangle\rangle_{0}+\langle\langle\tau_{m+2k_1-1}(\phi_\beta)\rangle\rangle_0
\\=&\frac{1}{2}\sum_{i=0}^{2k_1-2}\sum_{\alpha}(-1)^i\langle\langle\tau_i(\phi_\alpha)\tau_m(\phi_\beta)\rangle\rangle_0\langle\langle\tau_{2k_1-2-i}(\phi^\alpha)\rangle\rangle_0+\frac{1}{2}\sum_{i=0}^{2k_1-2}\sum_{\alpha}(-1)^i\langle\langle\tau_i(\phi_\alpha)\rangle\rangle_0\langle\langle\tau_{2k_1-2-i}(\phi^\alpha)\tau_m(\phi_\beta)\rangle\rangle_0   \\=&\sum_{i=0}^{2k_1-2}\sum_{\alpha}(-1)^i\langle\langle\tau_i(\phi_\alpha)\tau_m(\phi_\beta)\rangle\rangle_0\langle\langle\tau_{2k_1-2-i}(\phi^\alpha)\rangle\rangle_0.
\end{align*}
So we prove equation~$(i)$. Taking derivative of the above equation with respect to $t_n^\sigma$, we have
\begin{align*}
&\sum_{r,\alpha}\Tilde{t}_r^\alpha\langle\langle\tau_{r+2k_1-1}(\phi_\alpha)\tau_{m}(\phi_\beta)\tau_n(\phi_\sigma)\rangle\rangle_{0}
+\langle\langle\tau_{n+2k_1-1}(\phi_\sigma)\tau_m(\phi_\beta)\rangle\rangle_0+\langle\langle\tau_{m+2k_1-1}(\phi_\beta)\tau_n(\phi_\sigma)\rangle\rangle_0
\\=&\sum_{i=0}^{2k_1-2}\sum_{\alpha}(-1)^i\langle\langle\tau_i(\phi_\alpha)\tau_m(\phi_\beta)\tau_n(\phi_\sigma)\rangle\rangle_0\langle\langle\tau_{2k_1-2-i}(\phi^\alpha)\rangle\rangle_0
\\&+\sum_{i=0}^{2k_1-2}\sum_{\alpha}(-1)^i\langle\langle\tau_i(\phi_\alpha)\tau_m(\phi_\beta)\rangle\rangle_0\langle\langle\tau_{2k_1-2-i}(\phi^\alpha)\tau_n(\phi_\sigma)\rangle\rangle_0.
\end{align*}
Notice that by \cite[Lemma 2.1]{MR2775793}, 
\begin{align*}
&\langle\langle\tau_{n+2k_1-1}(\phi_\sigma)\tau_m(\phi_\beta)\rangle\rangle_0+\langle\langle\tau_{m+2k_1-1}(\phi_\beta)\tau_n(\phi_\sigma)\rangle\rangle_0
\\=&\sum_{i=0}^{2k_1-2}\sum_{\alpha}(-1)^i\langle\langle\tau_i(\phi_\alpha)\tau_m(\phi_\beta)\rangle\rangle_0\langle\langle\tau_{2k_1-2-i}(\phi^\alpha)\tau_n(\phi_\sigma)\rangle\rangle_0.
\end{align*}
So we obtain equation~$(ii)$
\begin{align*}
&\sum_{r,\alpha}\Tilde{t}_r^\alpha\langle\langle\tau_{r+2k_1-1}(\phi_\alpha)\tau_{m}(\phi_\beta)\tau_n(\phi_\sigma)\rangle\rangle_{0}
=\sum_{i=0}^{2k_1-2}\sum_{\alpha}(-1)^i\langle\langle\tau_i(\phi_\alpha)\tau_m(\phi_\beta)\tau_n(\phi_\sigma)\rangle\rangle_0\langle\langle\tau_{2k_1-2-i}(\phi^\alpha)\rangle\rangle_0.
\end{align*}
Taking derivative of equation~$(ii)$ with respect to $t_l^{\epsilon}$, we get 
\begin{align}\label{eqn:der-of-ii}
&\sum_{r,\alpha}\Tilde{t}_r^\alpha\langle\langle\tau_{r+2k_1-1}(\phi_\alpha)\tau_{m}(\phi_\beta)\tau_n(\phi_\sigma)\tau_l(\phi_\epsilon)\rangle\rangle_{0}
+\langle\langle\tau_{l+2k_1-1}(\phi_\epsilon)\tau_m(\phi_\beta)\tau_n(\phi_\sigma)\rangle\rangle_0
\nonumber\\=&\sum_{i=0}^{2k_1-2}\sum_{\alpha}(-1)^i\langle\langle\tau_i(\phi_\alpha)\tau_m(\phi_\beta)\tau_n(\phi_\sigma)\tau_{l}(\phi_\epsilon)\rangle\rangle_0\langle\langle\tau_{2k_1-2-i}(\phi^\alpha)\rangle\rangle_0
\nonumber\\&+\sum_{i=0}^{2k_1-2}\sum_{\alpha}(-1)^i\langle\langle\tau_i(\phi_\alpha)\tau_m(\phi_\beta)\tau_n(\phi_\sigma)\rangle\rangle_0\langle\langle\tau_{2k_1-2-i}(\phi^\alpha)\tau_{l}(\phi_\epsilon)\rangle\rangle_0.
\end{align}
By the generalized genus-0 topological recursion relation \cite[Lemma 2.2]{MR2775793}, we have
\begin{align}\label{eqn:generalized-g=0-trr}
&\langle\langle\tau_{l+2k_1-1}(\phi_\epsilon)\tau_m(\phi_\beta)\tau_n(\phi_\sigma)\rangle\rangle_0
=\sum_{i=0}^{2k_1-2}\sum_{\alpha}(-1)^i\langle\langle\tau_i(\phi_\alpha)\tau_m(\phi_\beta)\tau_n(\phi_\sigma)\rangle\rangle_0\langle\langle\tau_{2k_1-2-i}(\phi^\alpha)\tau_{l}(\phi_\epsilon)\rangle\rangle_0.  
\end{align}
Combining equations~\eqref{eqn:der-of-ii} and \eqref{eqn:generalized-g=0-trr}, we get
\begin{align*}
&\sum_{r,\alpha}\Tilde{t}_r^\alpha\langle\langle\tau_{r+2k_1-1}(\phi_\alpha)\tau_{m}(\phi_\beta)\tau_n(\phi_\sigma)\tau_l(\phi_\epsilon)\rangle\rangle_{0}
\\=&\sum_{i=0}^{2k_1-2}\sum_{\alpha}(-1)^i\langle\langle\tau_i(\phi_\alpha)\tau_m(\phi_\beta)\tau_n(\phi_\sigma)\tau_{l}(\phi_\epsilon)\rangle\rangle_0\langle\langle\tau_{2k_1-2-i}(\phi^\alpha)\rangle\rangle_0.
\end{align*}
So we prove equation~$(iii)$.     
\end{proof}

\subsection{Vanishing of $\Psi^{\mathbb{E}}_{1,1;k_1}$ for $k_1>1$}

By formula~\eqref{eqn:formula-psi-gnk1E},  the vanishing of $\Psi^{\mathbb{E}}_{1,1;k_1}$ for $k_1>1$ is equivalent to 
\begin{align}\label{eqn:11k1needtoporve}
&-\sum_{j=0}^{1}\sum_{r,\alpha,\beta}\left(\Delta_{2k_1-1}\hat{e}_{2-j}(\alpha)\right)(r) \langle c_1(X)^{j}\cup\phi_\alpha,\phi^\beta\rangle\tilde{t}_r^\alpha \frac{\partial F_{1;\emptyset}^{\mathbb{E}}(\mathbf{t})}{\partial t^\beta_{2k_1+r-j}}\nonumber
\\&+\frac{1}{2}\sum_{g_1+g_2=1}\sum_{j=0}^{1}\sum_{s,\beta,\beta'}(-1)^{-s}\left(\Delta_{2k_1-1}\check{e}_{2-j}(\beta)\right)(-s-1)\langle c_1(X)^{j}\cup\phi^{\beta},\phi^{\beta'}\rangle
\frac{\partial F_{g_1;\emptyset}^{\mathbb{E}}(\mathbf{t})}{\partial t_s^\beta}\frac{\partial F_{g_2;\emptyset}^{\mathbb{E}}(\mathbf{t})}{\partial t_{2k_1-1-s-j}^{\beta'}} 
\nonumber
\\&+\frac{1}{2}\sum_{j=0}^{1}\sum_{s,\beta,\beta'}(-1)^{-s}\left(\Delta_{2k_1-1}\check{e}_{2-j}(\beta)\right)(-s-1)\langle c_1(X)^{j}\cup\phi^{\beta},\phi^{\beta'}\rangle
\frac{\partial^2 F_{0;\emptyset}^{\mathbb{E}}(\mathbf{t})}{\partial t_s^\beta\partial t_{2k_1-1-s-j}^{\beta'}}
\nonumber\\=&0.
\end{align}    
Here simple calculations show
\begin{align}\label{eqn:delta2k_1-1e0e1}
(\Delta_{2k_1-1}\hat{e}_{0}(1,\alpha))(r)
=0,
\quad (\Delta_{2k_1-1}\hat{e}_{1}(1,\alpha))(r)
=2(2k_1-1)
\end{align}
and
\begin{align}\label{eqn:delta2k_1-1e2}
&(\Delta_{2k_1-1}\hat{e}_{2}(1,\alpha))(r)
=2(2k_1-1)(b_\alpha+r+k_1).
\end{align}
and similar formulae hold for functions $\check{e}_i(1,\alpha)$ $i=0,1,2$. 

To prove equation~\eqref{eqn:11k1needtoporve}, we need to compute its left hand side term by term. 
\begin{lemma}\label{lem:qqk_1-1st-term}
The first term in the left hand side of equation~\eqref{eqn:11k1needtoporve} equals to 
\begin{align}\label{eqn:lem:qqk_1-1st-term}
&-\sum_{j=0}^{1}\sum_{r,\alpha,\beta}\left(\Delta_{2k_1-1}\hat{e}_{2-j}(1,\alpha)\right)(r) \langle c_1(X)^{j}\cup\phi_\alpha,\phi^\beta\rangle\tilde{t}_r^\alpha \frac{\partial F_{1;\emptyset}^{\mathbb{E}}(\mathbf{t})}{\partial t^\beta_{2k_1+r-j}}
\nonumber\\=&-2(2k_1-1)\sum_{j=0}^{1}\sum_{\alpha,\beta}e_{1-j}(k_1-1+b_\beta)\langle\langle\tau_{2k_1-1-j}(c_1(X)^{j}\cup\phi_\beta)\rangle\rangle_0\langle\langle\phi^\beta\rangle\rangle_1
\nonumber\\&-2(2k_1-1)\sum_{j=0}^{1}\sum_{i}\sum_{\alpha,\beta}(-1)^i e_{1-j}(i+b_\alpha-k_1+1)\langle\langle\tau_{i-j}(c_1(X)^{j}\cup\phi_\alpha)\phi_\beta\rangle\rangle_0\langle\langle\tau_{2k_1-2-i}(\phi^\alpha)\rangle\rangle_0\langle\langle\phi^\beta\rangle\rangle_1
\nonumber\\&-\frac{1}{12}(2k_1-1)\sum_{j=0}^{1}\sum_{i}\sum_{\alpha,\beta}(-1)^i e_{1-j}(i+b_\alpha+1-k_1)\langle\langle\tau_{i-j}(c_1(X)^j\cup\phi_\alpha)\phi_\beta\phi^\beta\rangle\rangle_0\langle\langle\tau_{2k_1-2-i}(\phi^\alpha)\rangle\rangle_0.
\end{align}
\end{lemma}
\begin{proof}
By equations~\eqref{eqn:delta2k_1-1e0e1} and \eqref{eqn:delta2k_1-1e2}, we have
\begin{align}\label{eqn:-sumj=01sumralphadelta2k_1-1e_2-j}
&-\sum_{j=0}^{1}\sum_{r,\alpha}\left(\Delta_{2k_1-1}\hat{e}_{2-j}(1,\alpha)\right)(r) \langle c_1(X)^{j}\cup\phi_\alpha,\phi^\beta\rangle\tilde{t}_r^\alpha \frac{\partial F_{1;\emptyset}^{\mathbb{E}}(\mathbf{t})}{\partial t^\beta_{2k_1+r-j}}\nonumber
\\=&-2(2k_1+1)\sum_{r,\alpha}(r+b_\alpha+k_1) \Tilde{t}_r^\alpha\langle\langle\tau_{r+2k_1}(\phi_\alpha)\rangle\rangle_1
-2(2k_1+1)\sum_{r,\alpha}\tilde{t}_r^\alpha\langle\langle\tau_{r+2k_1-1}(c_1(X)\cup\phi_\alpha)\rangle\rangle_1\nonumber
\\=&-2(2k_1+1)\sum_{r,\alpha,\beta}(r+b_\alpha+k_1) \Tilde{t}_r^\alpha\langle\langle\tau_{r+2k_1-1}(\phi_\alpha)\phi^\beta\rangle\rangle_0\langle\langle\phi_\beta\rangle\rangle_1\nonumber
\\&-\frac{1}{12}(2k_1+1)\sum_{r,\alpha,\beta}(r+b_\alpha+k_1) \Tilde{t}_r^\alpha\langle\langle\tau_{r+2k_1-1}(\phi_\alpha)\phi_\beta\phi^\beta\rangle\rangle_0\nonumber
\\&-2(2k_1+1)\sum_{r,\alpha,\beta}\tilde{t}_r^\alpha\langle\langle\tau_{r+2k_1-2}(c_1(X)\cup\phi_\alpha)\phi^\beta\rangle\rangle_0
\langle\langle\phi_\beta\rangle\rangle_1\nonumber
\\&-\frac{1}{12}(2k_1+1)\sum_{r,\alpha,\beta}\tilde{t}_r^\alpha\langle\langle\tau_{r+2k_1-2}(c_1(X)\cup\phi_\alpha)\phi_\beta\phi^\beta\rangle\rangle_0.
\end{align}
Here in the last equality, we used  the genus-1 topological recursion relation
\begin{align}\label{eqn:g=1TRR}
\langle\langle\tau_s(\phi_\beta)\rangle\rangle_1 =\sum_{\alpha}\langle\langle\tau_{s-1}(\phi_\beta)\phi^\alpha\rangle\rangle_{0}
\langle\langle\phi_\alpha\rangle\rangle_1
+\frac{1}{24}\sum_{\alpha}\langle\langle\tau_{s-1}(\phi_\beta)\phi_\alpha\phi^\alpha\rangle\rangle_0.
\end{align}
% By the genus-1 topological recursion relation~\eqref{eqn:g=1TRR}, the above is equal to
% \begin{align*}
% % \\=&2\sum_{i}(-1)^i\Big(\delta_i^0\mathcal{C}_{\alpha\beta}+(i+b_\alpha+b_\beta)\langle\langle\tau_i(\phi_\alpha)\phi_\beta\rangle\rangle_0+\langle\langle\tau_{i-1}(c_1(X)\cup\phi_\alpha)\phi_\beta\rangle\rangle_0\Big)\langle\langle\tau_{2k_1-2-i}(\phi^\alpha)\rangle\rangle_0\langle\langle\phi^\beta\rangle\rangle_1
% % \\&+2\sum_{\alpha,\beta}(k_1-1+b_\beta)\langle\langle\tau_{2k_1-1}(\phi_\beta)\rangle\rangle_0\langle\langle\phi^\beta\rangle\rangle_1
% % \\&-2\sum_{\alpha,\beta}(k_1-1+b_\beta)\sum_{i}(-1)^i\langle\langle\tau_i(\phi_\alpha)\phi_\beta\rangle\rangle_0
% % \langle\langle\tau_{2k_1-2-i}(\phi^\alpha)\rangle\rangle_0\langle\langle\phi^\beta\rangle\rangle_1
% % \\&+\frac{1}{12}\sum_{i=0}^{2k_1-2}(-1)^i(i+b_\alpha+1-k_1)\langle\langle\tau_i(\phi_\alpha)\phi_\beta\phi^\beta\rangle\rangle_0\langle\langle\tau_{2k_1-2-i}(\phi^\alpha)\rangle\rangle_0
% % \\&+\frac{1}{12}\sum_{i=0}^{2k_1-2}(-1)^{i}\langle\langle\tau_{i-1}(c_1(X)\cup\phi_\alpha)\phi_\beta\phi^\beta\rangle\rangle_0\langle\langle\tau_{2k_1-2-i}(\phi^\alpha)\rangle\rangle_0
% \end{align*}

 By equation~$(ii)$ in Lemma~\ref{lem:der-quasi-homo-eqn}, we have
\begin{align*}
\langle\langle\mathcal{X}\tau_{r+2k_1-1}(\phi_\alpha)\phi_\beta\rangle\rangle_0=&(r+2k_1-1+b_\alpha+b_\beta)\langle\langle\tau_{r+2k_1-1}(\phi_\alpha)\phi_\beta\rangle\rangle_0
+\langle\langle\tau_{r+2k_1-2}(c_1(X)\cup\phi_\alpha)\phi_\beta\rangle\rangle_0.
\end{align*}
Then multiplying $\Tilde{t}_r^\alpha\langle\langle\phi_\beta\rangle\rangle_1$ on the above equation and summing over $r$ and $\alpha$, $\beta$, we get
\begin{align*}
&\sum_{r,\alpha,\beta}(r+b_\alpha+k_1) \Tilde{t}_r^\alpha\langle\langle\tau_{r+2k_1-1}(\phi_\alpha)\phi^\beta\rangle\rangle_0\langle\langle\phi_\beta\rangle\rangle_1
+\sum_{r,\alpha,\beta}\tilde{t}_r^\alpha\langle\langle\tau_{r+2k_1-2}(c_1(X)\cup\phi_\alpha)\phi^\beta\rangle\rangle_0
\langle\langle\phi_\beta\rangle\rangle_1
\\=&\sum_{r,\alpha,\beta}\Tilde{t}_r^\alpha\langle\langle\mathcal{X}\tau_{r+2k_1-1}(\phi_\alpha)\phi_\beta\rangle\rangle_0\langle\langle\phi^\beta\rangle\rangle_1
-\sum_{r,\alpha,\beta}(k_1-1+b_\beta)\Tilde{t}_r^\alpha\langle\langle\tau_{r+2k_1-1}(\phi_\alpha)\phi_\beta\rangle\rangle_0\langle\langle\phi^\beta\rangle\rangle_1.
\end{align*}
By equations~$(i)$ and $(ii)$ in Lemma~\ref{lem:der-g=0-FP-eqn}, we get
\begin{align}\label{eqn:-sumj=01sumralphadelta2k_1-1e_2-j-13}
&\sum_{r,\alpha,\beta}(r+b_\alpha+k_1) \Tilde{t}_r^\alpha\langle\langle\tau_{r+2k_1-1}(\phi_\alpha)\phi^\beta\rangle\rangle_0\langle\langle\phi_\beta\rangle\rangle_1
+\sum_{r,\alpha,\beta}\tilde{t}_r^\alpha\langle\langle\tau_{r+2k_1-2}(c_1(X)\cup\phi_\alpha)\phi^\beta\rangle\rangle_0
\langle\langle\phi_\beta\rangle\rangle_1
\nonumber\\=&\sum_{i}\sum_{\alpha,\beta}(-1)^i\langle\langle\tau_i(\phi_\alpha)\mathcal{X}\phi_\beta\rangle\rangle_0\langle\langle\tau_{2k_1-2-i}(\phi^\alpha)\rangle\rangle_0\langle\langle\phi^\beta\rangle\rangle_1
\nonumber\\&+\sum_{\alpha,\beta}(k_1-1+b_\beta)\langle\langle\tau_{2k_1-1}(\phi_\beta)\rangle\rangle_0\langle\langle\phi^\beta\rangle\rangle_1
\nonumber\\&-\sum_{\alpha,\beta}(k_1-1+b_\beta)\sum_{i}(-1)^i\langle\langle\tau_i(\phi_\alpha)\phi_\beta\rangle\rangle_0
\langle\langle\tau_{2k_1-2-i}(\phi^\alpha)\rangle\rangle_0\langle\langle\phi^\beta\rangle\rangle_1
\nonumber\\=&\sum_{i}(-1)^i\Big(\delta_{i,0}\mathcal{C}_{\alpha\beta}+(i+b_\alpha+b_\beta)\langle\langle\tau_i(\phi_\alpha)\phi_\beta\rangle\rangle_0+\langle\langle\tau_{i-1}(c_1(X)\cup\phi_\alpha)\phi_\beta\rangle\rangle_0\Big)\langle\langle\tau_{2k_1-2-i}(\phi^\alpha)\rangle\rangle_0\langle\langle\phi^\beta\rangle\rangle_1
\nonumber\\&+\sum_{\alpha,\beta}(k_1-1+b_\beta)\langle\langle\tau_{2k_1-1}(\phi_\beta)\rangle\rangle_0\langle\langle\phi^\beta\rangle\rangle_1
\nonumber\\&-\sum_{\alpha,\beta}(k_1-1+b_\beta)\sum_{i}(-1)^i\langle\langle\tau_i(\phi_\alpha)\phi_\beta\rangle\rangle_0
\langle\langle\tau_{2k_1-2-i}(\phi^\alpha)\rangle\rangle_0\langle\langle\phi^\beta\rangle\rangle_1
\nonumber\\=&\sum_{\beta}\langle\langle\tau_{2k_1-2}(c_1(X)\cup\phi_\beta)\rangle\rangle_0\langle\langle\phi^\beta\rangle\rangle_1
+\sum_{\beta}(k_1-1+b_\beta)\langle\langle\tau_{2k_1-1}(\phi_\beta)\rangle\rangle_0\langle\langle\phi^\beta\rangle\rangle_1
\nonumber\\&+\sum_{i}\sum_{\alpha,\beta}(-1)^i\langle\langle\tau_{i-1}(c_1(X)\cup\phi_\alpha)\phi_\beta\rangle\rangle_0\langle\langle\tau_{2k_1-2-i}(\phi^\alpha)\rangle\rangle_0\langle\langle\phi^\beta\rangle\rangle_1
\nonumber\\&+\sum_{i}\sum_{\alpha,\beta}(-1)^i(-k_1+1+i+b_\alpha)\langle\langle\tau_i(\phi_\alpha)\phi_\beta\rangle\rangle_0
\langle\langle\tau_{2k_1-2-i}(\phi^\alpha)\rangle\rangle_0\langle\langle\phi^\beta\rangle\rangle_1.
\end{align}
Here we used equation $(ii)$ in Lemma~\ref{lem:der-quasi-homo-eqn} in the second equality.

By equation~$(iii)$ in Lemma~\ref{lem:der-quasi-homo-eqn}, we have 
\begin{align*}
&(r+b_\alpha+k_1) \langle\langle\tau_{r+2k_1-1}(\phi_\alpha)\phi_\beta\phi^\beta\rangle\rangle_0
+\langle\langle\tau_{r+2k_1-2}(c_1(X)\cup\phi_\alpha)\phi_\beta\phi^\beta\rangle\rangle_0    
\\=&\langle\langle\mathcal{X}\tau_{r+2k_1-1}(\phi_\alpha)\phi_\beta\phi^\beta\rangle\rangle_0
-(k_1-1+b_\beta+b^\beta-\frac{3-d}{2})\langle\langle\tau_{r+2k_1-1}(\phi_\alpha)\phi_\beta\phi^\beta\rangle\rangle_0
\\=&\langle\langle\mathcal{X}\tau_{r+2k_1-1}(\phi_\alpha)\phi_\beta\phi^\beta\rangle\rangle_0
-(k_1-\frac{3-d}{2})\langle\langle\tau_{r+2k_1-1}(\phi_\alpha)\phi_\beta\phi^\beta\rangle\rangle_0.
\end{align*} 
Here in the last equality, we used basic identity $b_\beta+b^\beta=1$. 
Then multiplying $\Tilde{t}_r^\alpha$ on the above equation and summing over $r$ and $\alpha, \beta$, we get
\begin{align*}
&\sum_{r,\alpha,\beta}(r+b_\alpha+k_1) \Tilde{t}_r^\alpha\langle\langle\tau_{r+2k_1-1}(\phi_\alpha)\phi_\beta\phi^\beta\rangle\rangle_0
+\sum_{r,\alpha,\beta}\tilde{t}_r^\alpha\langle\langle\tau_{r+2k_1-2}(c_1(X)\cup\phi_\alpha)\phi_\beta\phi^\beta\rangle\rangle_0  
\\=&\sum_{r,\alpha,\beta}\Tilde{t}_r^\alpha\langle\langle\mathcal{X}\tau_{r+2k_1-1}(\phi_\alpha)\phi_\beta\phi^\beta\rangle\rangle_0
-(k_1-\frac{3-d}{2})\sum_{r,\alpha,\beta}\Tilde{t}_r^\alpha\langle\langle\tau_{r+2k_1-1}(\phi_\alpha)\phi_\beta\phi^\beta\rangle\rangle_0.
\end{align*}
By equation~$(ii)$ and $(iii)$ in Lemma~\ref{lem:der-g=0-FP-eqn}, then by equation~$(iii)$ in Lemma~\ref{lem:der-quasi-homo-eqn},  we get
\begin{align}\label{eqn:-sumj=01sumralphadelta2k_1-1e_2-j-24}
&\sum_{r,\alpha,\beta}(r+b_\alpha+k_1) \Tilde{t}_r^\alpha\langle\langle\tau_{r+2k_1-1}(\phi_\alpha)\phi_\beta\phi^\beta\rangle\rangle_0
+\sum_{r,\alpha,\beta}\tilde{t}_r^\alpha\langle\langle\tau_{r+2k_1-2}(c_1(X)\cup\phi_\alpha)\phi_\beta\phi^\beta\rangle\rangle_0 \nonumber \\=&\sum_{i=0}^{2k_1-2}\sum_{\alpha,\beta}(-1)^i\langle\langle\tau_i(\phi_\alpha)\mathcal{X}\phi_\beta\phi^\beta\rangle\rangle_0\langle\langle\tau_{2k_1-2-i}(\phi^\alpha)\rangle\rangle_0\nonumber
\\&-(k_1-\frac{3-d}{2})\sum_{i=0}^{2k_1-2}\sum_{\alpha,\beta}(-1)^i\langle\langle\tau_i(\phi_\alpha)\phi_\beta\phi^\beta\rangle\rangle_0
\langle\langle\tau_{2k_1-2-i}(\phi^\alpha)\rangle\rangle_0\nonumber
\\=&\sum_{i=0}^{2k_1-2}\sum_{\alpha,\beta}(-1)^i(i+b_\alpha+1-\frac{3-d}{2})\langle\langle\tau_i(\phi_\alpha)\phi_\beta\phi^\beta\rangle\rangle_0\langle\langle\tau_{2k_1-2-i}(\phi^\alpha)\rangle\rangle_0\nonumber
\\&+\sum_{i=0}^{2k_1-2}\sum_{\alpha,\beta}(-1)^{i}\langle\langle\tau_{i-1}(c_1(X)\cup\phi_\alpha)\phi_\beta\phi^\beta\rangle\rangle_0\langle\langle\tau_{2k_1-2-i}(\phi^\alpha)\rangle\rangle_0\nonumber
\\&-(k_1-\frac{3-d}{2})\sum_{i=0}^{2k_1-2}\sum_{\alpha,\beta}(-1)^i\langle\langle\tau_i(\phi_\alpha)\phi_\beta\phi^\beta\rangle\rangle_0
\langle\langle\tau_{2k_1-2-i}(\phi^\alpha)\rangle\rangle_0
\nonumber\\=&\sum_{i=0}^{2k_1-2}\sum_{\alpha,\beta}(-1)^i(i+b_\alpha+1-k_1)\langle\langle\tau_i(\phi_\alpha)\phi_\beta\phi^\beta\rangle\rangle_0\langle\langle\tau_{2k_1-2-i}(\phi^\alpha)\rangle\rangle_0\nonumber
\\&+\sum_{i=0}^{2k_1-2}\sum_{\alpha,\beta}(-1)^{i}\langle\langle\tau_{i-1}(c_1(X)\cup\phi_\alpha)\phi_\beta\phi^\beta\rangle\rangle_0\langle\langle\tau_{2k_1-2-i}(\phi^\alpha)\rangle\rangle_0.
\end{align}
Thus equation~\eqref{eqn:lem:qqk_1-1st-term} follows from equations~\eqref{eqn:-sumj=01sumralphadelta2k_1-1e_2-j}, \eqref{eqn:-sumj=01sumralphadelta2k_1-1e_2-j-13} and \eqref{eqn:-sumj=01sumralphadelta2k_1-1e_2-j-24}.
% so we have
% \begin{align*}
% &-\sum_{j=0}^{1}\sum_{r,\alpha}\left(\Delta_{2k_1-1}\hat{e}_{2-j}(1,\alpha)\right)(r) \cup\langle c_1(X)^{j}\cdot\phi_\alpha,\phi^\beta\rangle\cdot\tilde{t}_r^\alpha \frac{\partial F_{1;\emptyset}^{\mathbb{E}}(\mathbf{t})}{\partial t^\beta_{(2k_1-1)+r+1-j}}
% \\=&-2(2k_1-1)\sum_{i}(-1)^i\Big(\delta_i^0\mathcal{C}_{\alpha\beta}+(i+b_\alpha+b_\beta)\langle\langle\tau_i(\phi_\alpha)\phi_\beta\rangle\rangle_0+\langle\langle\tau_{i-1}(c_1\cdot\phi_\alpha)\phi_\beta\rangle\rangle_0\Big)\langle\langle\tau_{2k_1-2-i}(\phi^\alpha)\rangle\rangle_0\langle\langle\phi^\beta\rangle\rangle_1
% \\&-2(2k_1-1)\sum_{\alpha,\beta}(k_1-1+b_\beta)\langle\langle\tau_{2k_1-1}(\phi_\beta)\rangle\rangle_0\langle\langle\phi^\beta\rangle\rangle_1
% \\&+2(2k_1-1)\sum_{\alpha,\beta}(k_1-1+b_\beta)\sum_{i}(-1)^i\langle\langle\tau_i(\phi_\alpha)\phi_\beta\rangle\rangle_0
% \langle\langle\tau_{2k_1-2-i}(\phi^\alpha)\rangle\rangle_0\langle\langle\phi^\beta\rangle\rangle_1
% \\&-\frac{1}{12}(2k_1-1)\sum_{i=0}^{2k_1-2}(-1)^i(i+b_\alpha+1-k_1)\langle\langle\tau_i(\phi_\alpha)\phi_\beta\phi^\beta\rangle\rangle_0\langle\langle\tau_{2k_1-2-i}(\phi^\alpha)\rangle\rangle_0
% \\&-\frac{1}{12}(2k_1-1)\sum_{i=0}^{2k_1-2}(-1)^{i}\langle\langle\tau_{i-1}(c_1\cdot\phi_\alpha)\phi_\beta\phi^\beta\rangle\rangle_0\langle\langle\tau_{2k_1-2-i}(\phi^\alpha)\rangle\rangle_0
% \end{align*}
\end{proof}
\begin{lemma}\label{lem:qqk_1-2nd-term}
The second term in the left hand side of equation~\eqref{eqn:11k1needtoporve} equals to  
\begin{align}\label{eqn:lem:qqk_1-2nd-term}
&\frac{1}{2}\sum_{g_1+g_2=1}\sum_{j=0}^{1}\sum_{s,\beta,\beta'}(-1)^{-s}\left(\Delta_{2k_1-1}\check{e}_{2-j}(1,\beta)\right)(-s-1)\langle c_1(X)^{j}\cup\phi^{\beta},\phi^{\beta'}\rangle
\frac{\partial F_{g_1;\emptyset}^{\mathbb{E}}(\mathbf{t})}{\partial t_s^\beta}\frac{\partial F_{g_2;\emptyset}^{\mathbb{E}}(\mathbf{t})}{\partial t_{2k_1-1-s-j}^{\beta'}} 
\nonumber
\\=&2(2k_1-1)\sum_{\beta}\sum_{j=0}^{1}e_{1-j}(k_1-1+b^\beta)
\langle\langle\phi_\beta\rangle\rangle_{1}
\langle\langle\tau_{2k_1-1-j}(c_1(X)^{j}\cup\phi^\beta)\rangle\rangle_{0}
\nonumber\\&+2(2k_1-1)\sum_{s,\alpha,\beta}\sum_{j=0}^{1}(-1)^{-s}e_{1-j}(k_1-1-s+b^\beta)
\langle\langle\tau_{s-1}(\phi_\beta)\phi^\alpha\rangle\rangle_{0}
\langle\langle\phi_\alpha\rangle\rangle_1
\langle\langle\tau_{2k_1-1-s-j}(c_1(X)^j\cup\phi^\beta)\rangle\rangle_{0}
\nonumber\\&+\frac{1}{12}(2k_1-1)\sum_{s,\alpha,\beta}\sum_{j=0}^{1}(-1)^{-s}e_{1-j}(k_1-1-s+b^\beta)
\langle\langle\tau_{s-1}(\phi_\beta)\phi_\alpha\phi^\alpha\rangle\rangle_0\langle\langle\tau_{2k_1-1-s-j}(c_1(X)^{j}\cup\phi^\beta)\rangle\rangle_{0}. 
\end{align}
\end{lemma}
\begin{proof}

It is easy to see $j=0$ part in the first term on the left hand side of equation~\eqref{eqn:11k1needtoporve} equals to 
\begin{align*}
&\frac{1}{2}\sum_{g_1+g_2=g}\sum_{s,\beta,\beta'}(-1)^{-s}\left(\Delta_{2k_1-1}\check{e}_{2}(1,\beta)\right)(-s-1)\langle \phi^{\beta},\phi^{\beta'}\rangle
\frac{\partial F_{g_1;\emptyset}^{\mathbb{E}}(\mathbf{t})}{\partial t_s^\beta}\frac{\partial F_{g_2;\emptyset}^{\mathbb{E}}(\mathbf{t})}{\partial t_{2k_1-1-s}^{\beta'}}
\\=&\frac{1}{2}(2k_1-1)\sum_{g_1+g_2=1}(-1)^{-s}\sum_{s,\beta}(2k_1-2-2s+2b^\beta)
\langle\langle\tau_s(\phi_\beta)\rangle\rangle_{g_1}
\langle\langle\tau_{2k_1-1-s}(\phi^\beta)\rangle\rangle_{g_2}
\\=&2(2k_1-1)\sum_{\beta}(k_1-1+b^\beta)
\langle\langle\tau_0(\phi_\beta)\rangle\rangle_{1}
\langle\langle\tau_{2k_1-1}(\phi^\beta)\rangle\rangle_{0}
\\&+2(2k_1-1)\sum_{s\geq1,\beta}(-1)^{-s}(k_1-1-s+b^\beta)
\langle\langle\tau_s(\phi_\beta)\rangle\rangle_{1}
\langle\langle\tau_{2k_1-1-s}(\phi^\beta)\rangle\rangle_{0}.
\end{align*}

Applying the genus-1 topological recursion relation~\eqref{eqn:g=1TRR} to the last term, we get that it is  equal to
\begin{align}\label{eqn:j=0part=second-term-resu}
&2(2k_1-1)\sum_{\beta}(k_1-1+b^\beta)
\langle\langle\tau_0(\phi_\beta)\rangle\rangle_{1}
\langle\langle\tau_{2k_1-1}(\phi^\beta)\rangle\rangle_{0}\nonumber
\\&+2(2k_1-1)\sum_{s\geq1,\alpha,\beta}(-1)^{-s}(k_1-1-s+b^\beta)
\Big(\langle\langle\tau_{s-1}(\phi_\beta)\phi^\alpha\rangle\rangle_{0}
\langle\langle\phi_\alpha\rangle\rangle_1
\Big)\langle\langle\tau_{2k_1-1-s}(\phi^\beta)\rangle\rangle_{0}\nonumber
\\&+2(2k_1-1)\sum_{s\geq1,\alpha,\beta}(-1)^{-s}(k_1-1-s+b^\beta)
\Big(\frac{1}{24}\langle\langle\tau_{s-1}(\phi_\beta)\phi_\alpha\phi^\alpha\rangle\rangle_0\Big)\langle\langle\tau_{2k_1-1-s}(\phi^\beta)\rangle\rangle_{0}.    
\end{align}

And $j=1$ part in the first term in the left hand side of equation~\eqref{eqn:11k1needtoporve}  equals to 
\begin{align*}
&\frac{1}{2}\sum_{g_1+g_2=g}\sum_{s,\beta,\beta'}(-1)^{-s}\left(\Delta_{2k_1-1}\check{e}_{1}(1,\beta)\right)(-s-1)\langle c_1(X)\cup\phi^{\beta},\phi^{\beta'}\rangle
\frac{\partial F_{g_1;\emptyset}^{\mathbb{E}}(\mathbf{t})}{\partial t_s^\beta}\frac{\partial F_{g_2;\emptyset}^{\mathbb{E}}(\mathbf{t})}{\partial t_{2k_1-2-s}^{\beta'}} 
\\=&(2k_1-1)\sum_{g_1+g_2=1}(-1)^{-s}\sum_{s,\beta}
\langle\langle\tau_s(\phi_\beta)\rangle\rangle_{g_1}
\langle\langle\tau_{2k_1-2-s}(c_1(X)\cup\phi^\beta)\rangle\rangle_{g_2}
\\=&(2k_1-1)\sum_{s,\beta}(-1)^{-s}
\langle\langle\tau_s(\phi_\beta)\rangle\rangle_{0}
\langle\langle\tau_{2k_1-2-s}(c_1(X)\cup\phi^\beta)\rangle\rangle_{1}
\\&+(2k_1-1)\sum_{s,\beta}(-1)^{-s}\langle\langle\tau_s(\phi_\beta)\rangle\rangle_{1}
\langle\langle\tau_{2k_1-2-s}(c_1(X)\cup\phi^\beta)\rangle\rangle_{0}
\\=&2(2k_1-1)\sum_{s,\beta}(-1)^{-s}\langle\langle\tau_s(\phi_\beta)\rangle\rangle_{1}
\langle\langle\tau_{2k_1-2-s}(c_1(X)\cup\phi^\beta)\rangle\rangle_{0}.
\end{align*}
Applying the genus-1 topological recursion relation~\eqref{eqn:g=1TRR} to the last term, we get that it is  equal to
\begin{align}\label{eqn:j=1part=second-term-resu}
&2(2k_1-1)\sum_{\beta}
\langle\langle\tau_0(\phi_\beta)\rangle\rangle_{1}
\langle\langle\tau_{2k_1-2}(c_1(X)\cup\phi^\beta)\rangle\rangle_{0}\nonumber
\\&+2(2k_1-1)\sum_{s\geq1}\sum_{\alpha,\beta}(-1)^{-s}
\langle\langle\tau_{s-1}(\phi_\beta)\phi^\alpha\rangle\rangle_{0}
\langle\langle\phi_\alpha\rangle\rangle_1
\langle\langle\tau_{2k_1-2-s}(c_1(X)\cup\phi^\beta)\rangle\rangle_{0}\nonumber
\\&+2(2k_1-1)\sum_{s\geq1}\sum_{\alpha,\beta}
\frac{(-1)^{-s}}{24}\langle\langle\tau_{s-1}(\phi_\beta)\phi_\alpha\phi^\alpha\rangle\rangle_0\langle\langle\tau_{2k_1-2-s}(c_1(X)\cup\phi^\beta)\rangle\rangle_{0}.
\end{align}  
Combining the expressions~\eqref{eqn:j=0part=second-term-resu} and \eqref{eqn:j=1part=second-term-resu}, we prove this lemma.  
\end{proof}  
\begin{lemma}\label{lem:qqk_1-3rd-term}
The third term in the left hand side of equation~\eqref{eqn:11k1needtoporve} equals to 
\begin{align}\label{eqn:lem:qqk_1-3rd-term}
&\frac{1}{2}\sum_{s,\beta,\beta'}(-1)^{-s}\left(\Delta_{2k_1-1}\check{e}_{2}(1,\beta)\right)(-s-1)\langle \phi^{\beta},\phi^{\beta'}\rangle
\frac{\partial^2 F_{0;\emptyset}^{\mathbb{E}}(\mathbf{t})}{\partial t_s^\beta\partial t_{2k_1-1-s}^{\beta'}} 
\nonumber\\&+\frac{1}{2}\sum_{s,\beta,\beta'}(-1)^{-s}\left(\Delta_{2k_1-1}\check{e}_{1}(1,\beta)\right)(-s-1)\langle c_1(X)\cup\phi^{\beta},\phi^{\beta'}\rangle
\frac{\partial^2 F_{0;\emptyset}^{\mathbb{E}}(\mathbf{t})}{\partial t_s^\beta\partial t_{2k_1-2-s}^{\beta'}}
\nonumber\\=&(2k_1-1)\sum_{s,\beta}(-1)^{-s}(k_1-1-s+b^\beta)
\langle\langle\tau_s(\phi_\beta)\tau_{2k_1-1-s}(\phi^\beta)\rangle\rangle_0
\nonumber\\&+(2k_1-1)\sum_{s,\beta}(-1)^{-s}\langle\langle\tau_s(\phi_\beta)\tau_{2k_1-2-s}(c_1(X)\cup\phi^\beta)\rangle\rangle_0.
\end{align}    
\end{lemma}
\begin{proof}
This lemma follows from  simple calculations  
\begin{align*}
 (\Delta_{2k_1-1}\check{e}_{1}(1,\alpha))(-s-1)
=2(2k_1-1),\,
(\Delta_{2k_1-1}\check{e}_{2}(1,\alpha))(-s-1)
=2(2k_1-1)(b^\alpha-s-1+k_1).
\end{align*}
    
\end{proof}

By Lemma~\ref{lem:qqk_1-1st-term} and \ref{lem:qqk_1-2nd-term}, the sum of the first two terms in equation~\eqref{eqn:11k1needtoporve}
\begin{align}\label{eqn:sumof-first-two-terms}
&-\sum_{j=0}^{1}\sum_{r,\alpha,\beta}\left(\Delta_{2k_1-1}\hat{e}_{2-j}(\alpha)\right)(r) \langle c_1(X)^{j}\cup\phi_\alpha,\phi^\beta\rangle\tilde{t}_r^\alpha \frac{\partial F_{1;\emptyset}^{\mathbb{E}}(\mathbf{t})}{\partial t^\beta_{2k_1+r-j}}
\nonumber\\&+\frac{1}{2}\sum_{g_1+g_2=1}\sum_{j=0}^{1}\sum_{s,\beta,\beta'}(-1)^{-s}\left(\Delta_{2k_1-1}\check{e}_{2-j}(\beta)\right)(-s-1)\langle c_1(X)^{j}\cup\phi^{\beta},\phi^{\beta'}\rangle
\frac{\partial F_{g_1;\emptyset}^{\mathbb{E}}(\mathbf{t})}{\partial t_s^\beta}\frac{\partial F_{g_2;\emptyset}^{\mathbb{E}}(\mathbf{t})}{\partial t_{2k_1-1-s-j}^{\beta'}}
\nonumber\\=&0.
\end{align}    
In fact, it is easy to see the first term on the right hand side of \eqref{eqn:lem:qqk_1-1st-term} cancel with the first term on the right hand side of equation~\eqref{eqn:lem:qqk_1-2nd-term}, the second term on the right hand side of \eqref{eqn:lem:qqk_1-1st-term} cancel with the second term on the right hand side of equation~\eqref{eqn:lem:qqk_1-2nd-term} and the third term on the right hand side of \eqref{eqn:lem:qqk_1-1st-term} cancel with the third term on the right hand side of equation~\eqref{eqn:lem:qqk_1-2nd-term}.

Thus in order to prove the vanishing of $\Psi_{1,1;k_1}^{\mathbb{E}}$, via Lemma~\ref{lem:qqk_1-3rd-term}, we only need to prove the following identity 
\begin{align*}
&\sum_{s,\beta}(-1)^{-s}(k_1-1-s+b^\beta)
\langle\langle\tau_s(\phi_\beta)\tau_{2k_1-1-s}(\phi^\beta)\rangle\rangle_0
+\sum_{s,\beta}(-1)^{-s}\langle\langle\tau_s(\phi_\beta)\tau_{2k_1-2-s}(c_1(X)\cup\phi^\beta)\rangle\rangle_0=0.    
\end{align*}

For our purpose, we introduce $Q$ operator on the big phase space, which is defined by
\begin{align}
\label{eqn:Q-operator-definition}    Q(W):=\sum_{n,\alpha}f_{n,\alpha}\left((n+b_\alpha)\tau_{n}(\phi_\alpha) +\tau_{n-1}(c_1(X)\cup\phi_\alpha)\right)
\end{align}
for any vector fields $W=\sum_{n,\alpha}f_{n,\alpha}\tau_{n}(\phi_\alpha)$ on the big phase space.
% \Xin{geometrically, $b_\alpha$ comes from Hodge grading operator 
% \[G(\phi_\alpha)=b_\alpha\phi_\alpha,\, G(\phi^\alpha)=\sum_{\beta}\eta^{\alpha\beta}G(\phi_\beta)=\sum_{\beta}\eta^{\alpha\beta}b_\beta\phi_\beta=\sum_{\beta}\eta^{\alpha\beta}(1-b_\alpha)\phi_\beta
% =b^\alpha\phi^\alpha\]
% and
% \[b^\alpha\neq \sum_{\beta}\eta^{\alpha\beta}b_\beta\]}
By symmetry of indices, it is easy to see
\begin{align*}
\sum_{s,\alpha}(-1)^{s}\langle\langle\tau_s(\phi_\alpha)\tau_{2k-1-s}(\phi^\alpha)\rangle\rangle_0=0.    
\end{align*}

After simplification, it becomes (we always assume $k_1>1$)
\begin{align*}
&\sum_{s,\beta}(-1)^{-s}(k_1-1-s+b^\beta)
\langle\langle\tau_s(\phi_\beta)\tau_{2k_1-1-s}(\phi^\beta)\rangle\rangle_0
+\sum_{s,\beta}(-1)^{-s}\langle\langle\tau_s(\phi_\beta)\tau_{2k_1-2-s}(c_1(X)\cup\phi^\beta)\rangle\rangle_0
\\=&\sum_{s,\beta}(-1)^{-s}(2k_1-1-s+b^\beta)
\langle\langle\tau_s(\phi_\beta)\tau_{2k_1-1-s}(\phi^\beta)\rangle\rangle_0
+\sum_{s,\beta}(-1)^{-s}\langle\langle\tau_s(\phi_\beta)\tau_{2k_1-2-s}(c_1(X)\cup\phi^\beta)\rangle\rangle_0
\\=&\sum_{s,\beta}(-1)^{-s}
\langle\langle\tau_s(\phi_\beta)Q(\tau_{2k_1-1-s}(\phi^\beta))\rangle\rangle_0.
\end{align*} 
This can be proved 
to be 0 by the following lemma. 
\begin{lemma}\label{lem:-1jtau2l-1-jQtaujg=0} For $l\geq2$, 
\begin{align*}
\sum_{j=0}^{2l-1}\sum_{\alpha}(-1)^{j}\langle\langle\tau_{2l-1-j}(\phi^\alpha)Q(\tau_j(\phi_\alpha))\rangle\rangle_{0}=0.   \end{align*}   
\end{lemma}
\begin{proof}
By equation~$(ii)$ in Lemma~\ref{lem:der-quasi-homo-eqn},
\begin{align*}
&\sum_{j=0}^{2l-1}\sum_{\alpha}(-1)^{j}\langle\langle\tau_{2l-1-j}(\phi^\alpha)\mathcal{X} Q(\tau_j(\phi_\alpha))\rangle\rangle_0   
\\=&\sum_{j,\alpha}(-1)^{j}\langle\langle Q(\tau_{2l-1-j}(\phi^\alpha))Q(\tau_j(\phi_\alpha))\rangle\rangle_0+\sum_{j,\alpha}(-1)^{j}\langle\langle\tau_{2l-1-j}(\phi^\alpha)Q^2(\tau_j(\phi_\alpha))\rangle\rangle_0
\\&+\sum_{j,\alpha}(-1)^{j}\delta_{2l-1-j,0}\langle c_1(X)\cup\phi^\alpha,\pi Q(\tau_{j}(\phi_\alpha))\rangle
\\=&\sum_{j,\alpha}(-1)^{j}(2l-1-j+b^\alpha)\langle\langle\tau_{2l-1-j}(\phi^\alpha)Q(\tau_j(\phi_\alpha))\rangle\rangle_0
+\sum_{j,\alpha}(-1)^{j}\langle\langle\tau_{2l-2-j}(c_1(X)\cup\phi^\alpha)Q(\tau_j(\phi_\alpha))\rangle\rangle_0
\\&+\sum_{j,\alpha}(-1)^{j}(j+b_\alpha)\langle\langle\tau_{2l-1-j}(\phi^\alpha)Q(\tau_j(\phi_\alpha))\rangle\rangle_0
+\sum_{j,\alpha}(-1)^{j}\langle\langle\tau_{2l-1-j}(\phi^\alpha)Q(\tau_{j-1}(c_1(X)\cup\phi_\alpha))\rangle\rangle_0
\\&+\sum_{j,\alpha}(-1)^{j}\delta_{2l-1-j,0}\langle c_1(X)\cup\phi^\alpha,\pi Q(\tau_{j}(\phi_\alpha))\rangle
\\=&2l\sum_{j,\alpha}(-1)^{j}\langle\langle\tau_{2l-1-j}(\phi^\alpha)Q(\tau_{j}(\phi_\alpha))\rangle\rangle_0
\end{align*}
where $\pi$ is the projection from the big phase space to the small phase space, defined by $\pi(\tau_{n}(\phi_\alpha)):=\delta_{n,0}\phi_\alpha$.

By equation~\eqref{eqn:tau_jtau_m-j=TjTm-j}, we have
\begin{align*}
\sum_{j=0}^{2l-1}\sum_{\alpha}(-1)^{j}\langle\langle\tau_{2l-1-j}(\phi^\alpha)\mathcal{X} Q(\tau_j(\phi_\alpha))\rangle\rangle_0  =\sum_{j=0}^{2l-1}\sum_{\alpha}(-1)^{j}\langle\langle T^{2l-1-j}(\phi^\alpha)\mathcal{X} Q(T^j(\phi_\alpha))\rangle\rangle_0.      
\end{align*}
Recall the relations between $Q$ and $T$ operators
\begin{align*}
(QT^k-T^kQ)(W)=k T^k(W)-T^{k-1}(\mathcal{X}\bullet W)    
\end{align*}
for $k\geq1$.
Then this lemma follows from the genus-0 topological recursion relation of the form
\begin{align*}
\langle\langle T(\tau_{n}(\phi_\alpha))\tau_m(\phi_\beta)\tau_k(\phi_\gamma)\rangle\rangle_0=0
\end{align*}
for any $n,m,k$ and $\alpha,\beta,\gamma$. 
\end{proof}
\subsection{Vanishing of $\Psi^{\mathbb{E}}_{1,1;k_1}$ for $k_1=1$}
In this subsection, we prove the vanishing of $\Psi^{\mathbb{E}}_{1,1;1}$. By formula~\eqref{eqn:formula-psi-gnk1E}, let $\Psi^{\mathbb{E}}_{1,1;1}=\Psi^{\mathbb{E}'}_{1,1;1}+\Psi^{\mathbb{E}''}_{1,1;1}$, where 
\begin{align}\label{eqn:formula-psi-g=1n=1k1=1E-'}
&\Psi^{\mathbb{E}'}_{1,1;1}\nonumber   \\=&-\frac{B_{2k_1}}{(2k_1)!}\left\{\sum_{j=0}^{n+1}\sum_{r,\alpha,\beta}\left(\Delta_{2k_1-1}\hat{e}_{n+1-j}(n,\alpha)\right)(r) \langle c_1(X)^{j}\cup\phi_\alpha,\phi^\beta\rangle\tilde{t}_r^\alpha \frac{\partial F_{g;\emptyset}^{\mathbb{E}}(\mathbf{t})}{\partial t^\beta_{2k_1-1+r+n-j}}\nonumber
\right.\\&\left.-\frac{1}{2}\sum_{g_1+g_2=g}\sum_{j=0}^{n+1}\sum_{s,\beta,\beta'}(-1)^{-s}\left(\Delta_{2k_1-1}\check{e}_{n+1-j}(n,\beta)\right)(-s-1)\langle c_1(X)^{j}\cup\phi^{\beta},\phi^{\beta'}\rangle
\frac{\partial F_{g_1;\emptyset}^{\mathbb{E}}(\mathbf{t})}{\partial t_s^\beta}\frac{\partial F_{g_2;\emptyset}^{\mathbb{E}}(\mathbf{t})}{\partial t_{2k_1-2-s+n-j}^{\beta'}} 
\nonumber
\right.\\&\left.-\frac{1}{2}\sum_{j=0}^{n+1}\sum_{s,\beta,\beta'}(-1)^{-s}\left(\Delta_{2k_1-1}\check{e}_{n+1-j}(n,\beta)\right)(-s-1)\langle c_1(X)^{j}\cup\phi^{\beta},\phi^{\beta'}\rangle
\frac{\partial^2 F_{g-1;\emptyset}^{\mathbb{E}}(\mathbf{t})}{\partial t_s^\beta\partial t_{2k_1-2-s+n-j}^{\beta'}}\right\}\Bigg|_{\substack{g=1\\n=1\\k_1=1}}
\end{align}  
and
\begin{align}\label{eqn:formula-psi-g=1n=1k1=1E-''}
&\Psi^{\mathbb{E}''}_{1,1;1}\nonumber   \\=&\left\{\sum_{j=0}^{n+1}\sum_{r,\alpha,\beta}\left(\hat{e}_{n+1-j}(n,\alpha)\right)(r) \langle c_1(X)^{j}\cup\phi_\alpha,\phi^\beta\rangle\tilde{t}_r^\alpha \frac{\partial F_{g;k_1}^{\mathbb{E}}(\mathbf{t})}{\partial t^\beta_{r+n-j}}
\nonumber\right.\\&\left.-\frac{1}{2}\sum_{g_1+g_2=g}\sum_{j=0}^{n+1}\sum_{s,\beta,\beta'}(-1)^{-s}\left(\check{e}_{n+1-j}(n,\beta)\right)(-s-1)\langle c_1(X)^{j}\cup\phi^{\beta},\phi^{\beta'}\rangle
\frac{\partial F_{g_1;k_1}^{\mathbb{E}}(\mathbf{t})}{\partial t_s^\beta}\frac{\partial F_{g_2;\emptyset}^{\mathbb{E}}(\mathbf{t})}{\partial t_{-s-1+n-j}^{\beta'}} 
\nonumber
\right.\\&\left.-\frac{1}{2}\sum_{g_1+g_2=g}\sum_{j=0}^{n+1}\sum_{s,\beta,\beta'}(-1)^{-s}\left(\check{e}_{n+1-j}(n,\beta)\right)(-s-1)\langle c_1(X)^{j}\cup\phi^{\beta},\phi^{\beta'}\rangle
\frac{\partial F_{g_1;\emptyset}^{\mathbb{E}}(\mathbf{t})}{\partial t_s^\beta}\frac{\partial F_{g_2;k_1}^{\mathbb{E}}(\mathbf{t})}{\partial t_{-s-1+n-j}^{\beta'}} 
\nonumber
\right.\\&\left.-\frac{1}{2}\sum_{j=0}^{n+1}\sum_{s,\beta,\beta'}(-1)^{-s}\left(\check{e}_{n+1-j}(n,\beta)\right)(-s-1)\langle c_1(X)^{j}\cup\phi^{\beta},\phi^{\beta'}\rangle
\frac{\partial^2 F_{g-1;k_1}^{\mathbb{E}}(\mathbf{t})}{\partial t_s^\beta\partial t_{-s-1+n-j}^{\beta'}}\right\}\Bigg|_{\substack{g=1\\n=1\\k_1=1}}.
\end{align}  
\begin{lemma}\label{lem:psi-111E'}
\begin{align*}
\Psi^{\mathbb{E}'}_{1,1;1}    
=
\frac{1}{6}\sum_{\beta}b_\beta\langle\langle\tau_1(\phi_\beta)\phi^\beta\rangle\rangle_0+\frac{1}{12}\sum_{\beta}\langle\langle\phi_\beta(c_1(X)\cup\phi^\beta)\rangle\rangle_0.
\end{align*}
    
\end{lemma}
\begin{proof}
By equations~\eqref{eqn:lem:qqk_1-3rd-term},  \eqref{eqn:sumof-first-two-terms} and \eqref{eqn:formula-psi-g=1n=1k1=1E-'}, we have
\begin{align*}
&\Psi^{\mathbb{E}'}_{1,1;1}    
\\=&\frac{B_{2}}{2}\frac{1}{2}\sum_{j=0}^{n+1}\sum_{s,\beta,\beta'}(-1)^{-s}\left(\Delta_{2k_1-1}\check{e}_{n+1-j}(n,\beta)\right)(-s-1)\langle c_1(X)^{j}\cup\phi^{\beta},\phi^{\beta'}\rangle
\frac{\partial^2 F_{g-1;\emptyset}^{\mathbb{E}}(\mathbf{t})}{\partial t_s^\beta\partial t_{2k_1-2-s+n-j}^{\beta'}}
\\=&\frac{1}{12}\sum_{s,\beta}(-1)^{-s}(-s+b^\beta)\langle\langle\tau_{s}(\phi_\beta)\tau_{1-s}(\phi^\beta)\rangle\rangle_0
+\frac{1}{12}\sum_{\beta}\langle\langle\phi_\beta(c_1(X)\cup\phi^\beta)\rangle\rangle_0
\\=&\frac{1}{12}\sum_{\beta}b^\beta\langle\langle(\phi_\beta)\tau_{1}(\phi^\beta)\rangle\rangle_0
-\frac{1}{12}\sum_{\beta}(-1+b^\beta)\langle\langle\tau_{1}(\phi_\beta)(\phi^\beta)\rangle\rangle_0
+\frac{1}{12}\sum_{\beta}\langle\langle\phi_\beta(c_1(X)\cup\phi^\beta)\rangle\rangle_0
\\=&
\frac{1}{6}\sum_{\beta}b_\beta\langle\langle\tau_1(\phi_\beta)\phi^\beta\rangle\rangle_0+\frac{1}{12}\sum_{\beta}\langle\langle\phi_\beta(c_1(X)\cup\phi^\beta)\rangle\rangle_0.
\end{align*}
    
\end{proof}

\begin{lemma}\label{lem:psi-111E''}
\begin{align*}
\Psi_{1,1;1}^{\mathbb{E}''}
=&
-\frac{1}{6}\sum_{\alpha}b_\alpha 
\langle\langle\tau_1(\phi_\alpha)\phi^\alpha\rangle\rangle_0
-\frac{1}{12}\sum_{\beta}(2b_\beta+1)\langle\langle(c_1(X)\cup\phi_\beta)\phi^\beta\rangle\rangle_0.
\end{align*}    
\end{lemma}
\begin{proof}
Recall on $\overline{\mathcal{M}}_{1,1}$, there is a tautological relation
\begin{align*}
\ch_1(\mathbb{E})=\lambda_1 =\frac{1}{12}\delta_{irr}   
\end{align*}
where $\delta_{irr}$ is the boundary stratum with a non-separated nodal point. It can be translated to a universal equation for Hodge integrals of any target varieties
\begin{align}
\label{eqn:tauo-rel-univer-eqn-1,1}\langle\langle\tau_{n}(\phi_\alpha);\ch_1(\mathbb{E})\rangle\rangle_1=\frac{1}{24}\sum_{\beta}\langle\langle\tau_n(\phi_\alpha)\phi_\beta\phi^\beta\rangle\rangle_0.    
\end{align}
Then
\begin{align*}
&\Psi^{\mathbb{E}''}_{1,1;1}\nonumber
\\=&\sum_{j=0}^{2}\sum_{r,\alpha}\left(\hat{e}_{2-j}(1,\alpha)\right)(r) \tilde{t}_r^\alpha\langle\langle\tau_{r+1-j}(c_1(X)^j\cup\phi_\alpha);\ch_1(\mathbb{E})\rangle\rangle_1
+\sum_{\beta}b_\beta b^\beta
\langle\langle\phi_\beta;\ch_1(\mathbb{E})\rangle\rangle_1\langle\langle\phi^\beta\rangle\rangle_0\nonumber
\\=&\sum_{r,\alpha}(r+b_\alpha)(r+b_\alpha+1)\tilde{t}_r^\alpha\langle\langle\tau_{r+1}(\phi_\alpha);\ch_1(\mathbb{E})\rangle\rangle_1
+\sum_{r,\alpha}(2r+2b_\alpha+1)\tilde{t}_r^\alpha\langle\langle\tau_{r}(c_1(X)\cup\phi_\alpha);\ch_1(\mathbb{E})\rangle\rangle_1
\nonumber
\\&+\sum_{r,\alpha}\tilde{t}_r^\alpha\langle\langle\tau_{r-1}(c_1(X)^2\cup\phi_\alpha);\ch_1(\mathbb{E})\rangle\rangle_1
+\sum_{\beta}b_\beta b^\beta
\langle\langle\phi_\beta;\ch_1(\mathbb{E})\rangle\rangle_1\langle\langle\phi^\beta\rangle\rangle_0\nonumber
\\=&\frac{1}{24}\sum_{r,\alpha,\beta}(r+b_\alpha)(r+b_\alpha+1)\tilde{t}_r^\alpha\langle\langle\tau_{r+1}(\phi_\alpha)\phi_\beta\phi^\beta\rangle\rangle_0
+\frac{1}{24}\sum_{r,\alpha,\beta}(2r+2b_\alpha+1)\tilde{t}_r^\alpha\langle\langle\tau_{r}(c_1(X)\cup\phi_\alpha)\phi_\beta\phi^\beta\rangle\rangle_0
\nonumber
\\&+\frac{1}{24}\sum_{r,\alpha,\beta}\tilde{t}_r^\alpha\langle\langle\tau_{r-1}(c_1(X)^2\cup\phi_\alpha)\phi_\beta\phi^\beta\rangle\rangle_0
\nonumber
+\frac{1}{24}\sum_{\alpha,\beta}b_\alpha b^\alpha
\langle\langle\phi_\alpha\phi_\beta\phi^\beta\rangle\rangle_0\langle\langle\phi^\alpha\rangle\rangle_0\nonumber
\end{align*}  
where we used universal equation~\eqref{eqn:tauo-rel-univer-eqn-1,1} in the last equality. 
Recall the genus-0 $L_1$ constraint for descendant Gromov-Witten invariants gives us
\begin{align*}
&\sum_{r,\alpha}(r+b_\alpha)(r+b_\alpha+1)\tilde{t}_r^\alpha\langle\langle\tau_{r+1}(\phi_\alpha)\rangle\rangle_0
+\sum_{r,\alpha}(2r+2b_\alpha+1)\tilde{t}_r^\alpha\langle\langle\tau_{r}(c_1(X)\cup\phi_\alpha)\rangle\rangle_0\nonumber
\\&
+\sum_{r,\alpha}\tilde{t}_r^\alpha\langle\langle\tau_{r-1}(c_1(X)^2\cup\phi_\alpha)\rangle\rangle_0
+\frac{1}{2}\sum_{\alpha}b_\alpha b^\alpha
\langle\langle\phi_\alpha\rangle\rangle_0\langle\langle\phi^\alpha\rangle\rangle_0\nonumber  
+\frac{1}{2}\sum_{\alpha,\beta}(\mathcal{C}^2)_{\alpha\beta}t_0^\alpha t_0^\beta\nonumber
\\&=0.
\end{align*}
Taking derivative along $\frac{\partial}{\partial t_0^\beta}$, we get
\begin{align*}
&\sum_{r,\alpha}(r+b_\alpha)(r+b_\alpha+1)\tilde{t}_r^\alpha\langle\langle\tau_{r+1}(\phi_\alpha)\phi_\beta\rangle\rangle_0
+b_\beta(b_\beta+1)\langle\langle\tau_{1}(\phi_\beta)\rangle\rangle_0
\\&+\sum_{r,\alpha}(2r+2b_\alpha+1)\tilde{t}_r^\alpha\langle\langle\tau_{r}(c_1(X)\cup\phi_\alpha)\phi_\beta\rangle\rangle_0+(2b_\beta+1)\langle\langle(c_1(X)\cup\phi_\beta)\rangle\rangle_0
\\&+\sum_{r,\alpha}\tilde{t}_r^\alpha\langle\langle\tau_{r-1}(c_1(X)^2\cup\phi_\alpha)\phi_\beta\rangle\rangle_0
+\frac{1}{2}\sum_{\alpha}b_\alpha b^\alpha
\langle\langle\phi_\alpha\phi_\beta\rangle\rangle_0\langle\langle\phi^\alpha\rangle\rangle_0
+\frac{1}{2}\sum_{\alpha}b_\alpha b^\alpha
\langle\langle\phi_\alpha\rangle\rangle_0\langle\langle\phi^\alpha\phi_\beta\rangle\rangle_0
\\&+\sum_{\alpha}(\mathcal{C}^2)_{\alpha\beta}t_0^\alpha \nonumber
\\&=0.
\end{align*}
Taking derivative along $\frac{\partial}{\partial t_0^\sigma}$, we get
\begin{align*}
&\sum_{r,\alpha}(r+b_\alpha)(r+b_\alpha+1)\tilde{t}_r^\alpha\langle\langle\tau_{r+1}(\phi_\alpha)\phi_\beta\phi_\sigma\rangle\rangle_0
+b_\sigma(b_\sigma+1)\langle\langle\tau_{1}(\phi_\sigma)\phi_\beta\rangle\rangle_0
+b_\beta(b_\beta+1)\langle\langle\tau_{1}(\phi_\beta)\phi_\sigma\rangle\rangle_0
\\&+\sum_{r,\alpha}(2r+2b_\alpha+1)\tilde{t}_r^\alpha\langle\langle\tau_{r}(c_1(X)\cup\phi_\alpha)\phi_\beta\phi_\sigma\rangle\rangle_0+(2b_\sigma+1)\langle\langle(c_1(X)\cup\phi_\sigma)\phi_\beta\rangle\rangle_0
\\&+(2b_\beta+1)\langle\langle(c_1(X)\cup\phi_\beta)\phi_\sigma\rangle\rangle_0
+\sum_{r,\alpha}\tilde{t}_r^\alpha\langle\langle\tau_{r-1}(c_1(X)^2\cup\phi_\alpha)\phi_\beta\phi_\sigma\rangle\rangle_0
+\sum_{\alpha}b_\alpha b^\alpha
\langle\langle\phi_\alpha\phi_\beta\phi_\sigma\rangle\rangle_0\langle\langle\phi^\alpha\rangle\rangle_0
\\&+\sum_{\alpha}b_\alpha b^\alpha
\langle\langle\phi_\alpha\phi_\beta\rangle\rangle_0\langle\langle\phi^\alpha\phi_\sigma\rangle\rangle_0
+(\mathcal{C}^2)_{\sigma\beta} \nonumber
\\&=0.
\end{align*}
Multiplying $\eta^{\sigma\beta}$ and summing over $\sigma$, we get
\begin{align*}
&\sum_{r,\alpha}(r+b_\alpha)(r+b_\alpha+1)\tilde{t}_r^\alpha\langle\langle\tau_{r+1}(\phi_\alpha)\phi_\beta\phi^\beta\rangle\rangle_0
+b^\beta(b^\beta+1)\langle\langle\tau_{1}(\phi^\beta)\phi_\beta\rangle\rangle_0
+b_\beta(b_\beta+1)\langle\langle\tau_{1}(\phi_\beta)\phi^\beta\rangle\rangle_0
\\&+\sum_{r,\alpha}(2r+2b_\alpha+1)\tilde{t}_r^\alpha\langle\langle\tau_{r}(c_1(X)\cup\phi_\alpha)\phi_\beta\phi^\sigma\rangle\rangle_0+(2b^\beta+1)\langle\langle(c_1(X)\cup\phi^\beta)\phi_\beta\rangle\rangle_0
\\&+(2b_\beta+1)\langle\langle(c_1(X)\cup\phi_\beta)\phi^\beta\rangle\rangle_0
+\sum_{r,\alpha}\tilde{t}_r^\alpha\langle\langle\tau_{r-1}(c_1(X)^2\cup\phi_\alpha)\phi_\beta\phi^\beta\rangle\rangle_0
+\sum_{\alpha}b_\alpha b^\alpha
\langle\langle\phi_\alpha\phi_\beta\phi^\beta\rangle\rangle_0\langle\langle\phi^\alpha\rangle\rangle_0\\&+\sum_{\alpha}b_\alpha b^\alpha
\langle\langle\phi_\alpha\phi_\beta\rangle\rangle_0\langle\langle\phi^\alpha\phi^\beta\rangle\rangle_0
\\&=0.
\end{align*}
After simplification and summing over $\beta$, we have
\begin{align*}
&\frac{1}{24}\sum_{r,\alpha,\beta}(r+b_\alpha)(r+b_\alpha+1)\tilde{t}_r^\alpha\langle\langle\tau_{r+1}(\phi_\alpha)\phi_\beta\phi^\beta\rangle\rangle_0
+\frac{1}{24}\sum_{r,\alpha,\beta}(2r+2b_\alpha+1)\tilde{t}_r^\alpha\langle\langle\tau_{r}(c_1(X)\cup\phi_\alpha)\phi_\beta\phi^\sigma\rangle\rangle_0
\\&+\frac{1}{24}\sum_{r,\alpha,\beta}\tilde{t}_r^\alpha\langle\langle\tau_{r-1}(c_1(X)^2\cup\phi_\alpha)\phi_\beta\phi^\beta\rangle\rangle_0
+\frac{1}{24}\sum_{\alpha,\beta}b_\alpha b^\alpha
\langle\langle\phi_\alpha\phi_\beta\phi^\beta\rangle\rangle_0\langle\langle\phi^\alpha\rangle\rangle_0
\\=&-\frac{1}{6}\sum_{\beta}b_\beta\langle\langle\tau_{1}(\phi_\beta)\phi^\beta\rangle\rangle_0
-\frac{1}{12}\sum_{\beta}(2b_\beta+1)\langle\langle(c_1(X)\cup\phi_\beta)\phi^\beta\rangle\rangle_0.
\end{align*}    
\end{proof}
By Lemma~\ref{lem:psi-111E'} and Lemma~\ref{lem:psi-111E''}, we have
\begin{align*}
\Psi_{1,1;1}^{\mathbb{E}}=\Psi_{1,1;1}^{\mathbb{E}'}+\Psi_{1,1;1}^{\mathbb{E}''}= -\frac{1}{6}\sum_{\alpha}b_\alpha\langle\langle(c_1(X)\cup\phi_\alpha)\phi^\alpha\rangle\rangle_0.   
\end{align*}
Notice that $\mathcal{C}_\alpha^\beta\neq0$ implies $b_\beta=b_\alpha+1$, so
\begin{align*}
\Psi_{1,1;1}^{\mathbb{E}}=& -\frac{1}{6}\sum_{\alpha,\beta}b_\alpha\mathcal{C}_\alpha^\beta\langle\langle\phi_\beta\phi^\alpha\rangle\rangle_0=-\frac{1}{6}\sum_{\alpha,\beta}(b_\beta-1)\mathcal{C}_\alpha^\beta\langle\langle\phi_\beta\phi^\alpha\rangle\rangle_0
\\=&\frac{1}{6}\sum_{\beta}b^\beta\langle\langle\phi_\beta(c_1(X)\cup\phi^\beta)\rangle\rangle_0=\frac{1}{6}\sum_{\beta}b_\beta\langle\langle\phi^\beta(c_1(X)\cup\phi_\beta)\rangle\rangle_0=0.
\end{align*}

\section{Proof of Theorem~\ref{thm:psiEg1k1}}\label{sec:proof-thm-g>=2}
\subsection{$\mathcal{L}_0$ vector field}

There are two important vector fields on the big phase space, which are crucial in our proof of Theorem~\ref{thm:psiEg1k1}. One is the dilaton vector field 
\begin{align*}
\mathcal{D}=-\sum_{n,\alpha}\tilde{t}_n^\alpha\tau_n(\phi_\alpha)    
\end{align*}
which satisfies the dilaton equation
\begin{align}\label{eqn:genus-g-dilaton-equation}
\langle\langle\mathcal{D}\rangle\rangle_g = (2g-2)F_g +\frac{\chi(X)}{24}\delta_{g,1}. 
\end{align}
Another one is $\mathcal{L}_0$ vector field
\begin{align*}
\mathcal{L}_0=\sum_{n,\alpha}(n+b_\alpha)\tilde{t}_n^\alpha\tau_n(\phi_\alpha)+\sum_{n,\alpha}\tilde{t}_n^\alpha\tau_{n-1}(c_1(X)\cup\phi_\alpha)    
\end{align*}
which satisfies Virasoro $L_0$ equation, 
\begin{align}\label{eqn:genus-g-L_0-equation}
\langle\langle\mathcal{L}_0\rangle\rangle_g=-\frac{1}{2}\delta_{g,0} \sum_{\alpha,\beta}\mathcal{C}_{\alpha\beta}t_0^\alpha t_0^\beta-\frac{1}{24}\delta_{g,1}\left(\frac{3-d}{2}\int_{X}c_d(X)-\int_{X}c_1(X)\cup c_{d-1}(X)\right)
.\end{align}
This equation was first discovered by Hori (cf. \cite{MR1323382}).
\begin{lemma}\label{lem:T2lL0genus-g}
For $l>0$, 
\begin{align}\label{eqn:lem:T2lL0genus-g}
 T^{2l}(\mathcal{L}_0)
=\tau_+^{2l}(\mathcal{L}_0)+\sum_{i=0}^{2l-1}\sum_{\alpha}(-1)^i\Big(\langle\langle Q(\tau_i(\phi_\alpha))\rangle\rangle_0+\delta_{i,0}\sum_{\beta}\mathcal{C}_{\alpha\beta}t_0^\beta\Big)\tau_{2l-1-i}(\phi^\alpha)
\end{align}    
and
\begin{align}\label{eqn:lem:T2lDgenus-g}
T^{2l}(\mathcal{D})
=\tau_{+}^{2l}(\mathcal{D})+\sum_{i=0}^{2l-1}\sum_{\alpha}(-1)^i
\langle\langle\tau_i(\phi^\alpha)\rangle\rangle_0
\tau_{2l-1-i}(\phi_\alpha).    
\end{align}
\end{lemma}
\begin{proof}
By equation~\eqref{eqn:T^k(W)-tau_kW}, 
\begin{align*}
&T^{2l}(\mathcal{L}_0)
=\tau_{+}^{2l}(\mathcal{L}_0)-\sum_{i=0}^{2l-1}\sum_{\alpha}(-1)^i
\langle\langle\mathcal{L}_0\tau_i(\phi^\alpha)\rangle\rangle_0
\tau_{2l-1-i}(\phi_\alpha).
\end{align*}
% and 
% \begin{align*}
% \mathcal{L}_0=-\mathcal{X}-\frac{3-d}{2}\mathcal{D}    
% \end{align*}
Recall the genus-0 $L_0$ constraint is
\begin{align*}
\langle\langle\mathcal{L}_0\rangle\rangle_0 =-\frac{1}{2}\sum_{\alpha,\beta}\mathcal{C}_{\alpha\beta}t_0^\alpha t_0^\beta.   
\end{align*}
Taking derivative along $\frac{\partial}{\partial t_i^\alpha}$, we get
\begin{align*}
 &\langle\langle\mathcal{L}_0\tau_i(\phi_\alpha)\rangle\rangle_0  
 % =-\langle\langle\mathcal{X}\tau_i(\phi_\alpha)\rangle\rangle_0-\frac{3-d}{2}\langle\langle\mathcal{D}\tau_i(\phi_\alpha)\rangle\rangle_0
% \\=&-(i+b_\alpha+\frac{3-d}{2})\langle\langle\tau_i(\phi_\alpha)\rangle\rangle_0-\langle\langle\tau_{i-1}(c_1(X)\cup\phi_\alpha)\rangle\rangle_0-\delta_{i}^{0}\sum_{\beta}\mathcal{C}_{\alpha\beta}t_0^\beta+\frac{3-d}{2}\langle\langle\tau_{i}(\phi_\alpha)\rangle\rangle_0
=-(i+b_\alpha)\langle\langle\tau_i(\phi_\alpha)\rangle\rangle_0-\langle\langle\tau_{i-1}(c_1(X)\cup\phi_\alpha)\rangle\rangle_0-\delta_{i,0}\sum_{\beta}\mathcal{C}_{\alpha\beta}t_0^\beta.
\end{align*}
Together with the definition of $Q$ operator, we prove equation~\eqref{eqn:lem:T2lL0genus-g}. Similarly, by equation~\eqref{eqn:T^k(W)-tau_kW},
\begin{align}\label{eqn:T2lD11}
&T^{2l}(\mathcal{D})
=\tau_{+}^{2l}(\mathcal{D})-\sum_{i=0}^{2l-1}\sum_{\alpha}(-1)^i
\langle\langle\mathcal{D}\tau_i(\phi^\alpha)\rangle\rangle_0
\tau_{2l-1-i}(\phi_\alpha).
\end{align}
Recall the genus-0 dilaton equation 
\begin{align*}
\langle\langle\mathcal{D}\rangle\rangle_0   =-2F_0. 
\end{align*}
Taking derivative along $\frac{\partial}{\partial t_i^\alpha}$, we get
\begin{align}\label{eqn:der-g=0-dilaton}
\langle\langle\mathcal{D}\tau_i(\phi_\alpha)\rangle\rangle_0=- \langle\langle\tau_i(\phi_\alpha)\rangle\rangle_0.   
\end{align}
Plugging equation~\eqref{eqn:der-g=0-dilaton} into equation~\eqref{eqn:T2lD11}, we prove equation~\eqref{eqn:lem:T2lDgenus-g}.
% so 
% \begin{align*}
% &\langle\langle T^{2l}(\mathcal{L}_0)\rangle\rangle_g
% \\=&\langle\langle\tau_+^{2l}(\mathcal{L}_0)\rangle\rangle_g+\sum_{i=0}^{2l-1}(-1)^i\Big((i+b_\alpha)\langle\langle\tau_i(\phi_\alpha)\rangle\rangle_0+\langle\langle\tau_{i-1}(c_1(X)\cup\phi_\alpha)\rangle\rangle_0+\delta_{i}^{0}\sum_{\beta}\mathcal{C}_{\alpha\beta}t_0^\beta\Big)\langle\langle\tau_{2l-1-i}(\phi^\alpha)\rangle\rangle_g
% \\=&\langle\langle\tau_+^{2l}(\mathcal{L}_0)\rangle\rangle_g+\sum_{i=0}^{2l-1}(-1)^i\Big(\langle\langle Q(\tau_i(\phi_\alpha))\rangle\rangle_0+\delta_{i}^{0}\sum_{\beta}\mathcal{C}_{\alpha\beta}t_0^\beta\Big)\langle\langle\tau_{2l-1-i}(\phi^\alpha)\rangle\rangle_g
% \end{align*}    
\end{proof}
\subsection{Vanishing of $\Psi^{\mathbb{E}}_{g,1;k_1}$} In this subsection, we prove the vanishing of $\Psi^{\mathbb{E}}_{g,1;k_1}$ for $g\geq2$ and 
$k_1\geq \max\{g+1,\frac{3g-1}{2}\}$.
Since $\ch_{2k-1}(\mathbb{E})=0$ on $\overline{\mathcal{M}}_{g,n}$ for $k>g$, by equation~\eqref{eqn:formula-psi-gnk1E}, we have
\begin{align*}
&\Psi^{\mathbb{E}}_{g,1;k_1}   
\\=&-\frac{B_{2k_1}}{(2k_1)!}\left\{-\sum_{j=0}^{2}\sum_{r,\alpha,\beta}\left(\Delta_{2k_1-1}\hat{e}_{2-j}(1,\alpha)\right)(r) \langle c_1(X)^{j}\cup\phi_\alpha,\phi^\beta\rangle\tilde{t}_r^\alpha \frac{\partial F_{g;\emptyset}^{\mathbb{E}}(\mathbf{t})}{\partial t^\beta_{2k_1+r-j}}
\right.\\&\left.+\frac{1}{2}\sum_{g_1+g_2=g}\sum_{j=0}^{2}\sum_{s,\beta,\beta'}(-1)^{-s}\left(\Delta_{2k_1-1}\check{e}_{2-j}(1,\beta)\right)(-s-1)\langle c_1(X)^{j}\cup\phi^{\beta},\phi^{\beta'}\rangle
\frac{\partial F_{g_1;\emptyset}^{\mathbb{E}}(\mathbf{t})}{\partial t_s^\beta}\frac{\partial F_{g_2;\emptyset}^{\mathbb{E}}(\mathbf{t})}{\partial t_{2k_1-1-s-j}^{\beta'}} 
\right.\\&\left.+\frac{1}{2}\sum_{j=0}^{2}\sum_{s,\beta,\beta'}(-1)^{-s}\left(\Delta_{2k_1-1}\check{e}_{2-j}(1,\beta)\right)(-s-1)\langle c_1(X)^{j}\cup\phi^{\beta},\phi^{\beta'}\rangle
\frac{\partial^2 F_{g-1;\emptyset}^{\mathbb{E}}(\mathbf{t})}{\partial t_s^\beta\partial t_{2k_1-1-s-j}^{\beta'}}\right\}
\end{align*}
here
\begin{align*}
(\Delta_{2k_1-1}\hat{e}_{0}(1,\alpha))(r)
=0,
\quad (\Delta_{2k_1-1}\hat{e}_{1}(1,\alpha))(r)
=2(2k_1-1)
\end{align*}
and
\begin{align*}
&(\Delta_{2k_1-1}\hat{e}_{2}(1,\alpha))(r)
=2(2k_1-1)(b_\alpha+r+k_1).
\end{align*}
and similar formulae hold for functions $\check{e}_i(1,\alpha)$ $i=0,1,2$.
So we only need to prove the following \begin{align}\label{eqn:sumnalphank1balpha}
&\sum_{r,\alpha}(r+k_1+b_\alpha)\tilde{t}_r^\alpha\langle\langle\tau_{r+2k_1}(\phi_\alpha)\rangle\rangle_g
+\sum_{r,\alpha}\tilde{t}_r^\alpha\langle\langle\tau_{r+2k_1-1}(c_1(X)\cup\phi_\alpha)\rangle\rangle_{g}
\nonumber\\&+\frac{1}{2}\sum_{g_1+g_2=g}\sum_{j=0}^{2k_1-1}\sum_{\alpha}(-1)^{j}\langle\langle\tau_{2k_1-1-j}(\phi^\alpha)\rangle\rangle_{g_1}\langle\langle Q(\tau_j(\phi_\alpha))\rangle\rangle_{g_2}
\nonumber\\&+\frac{1}{2}\sum_{j=0}^{2k_1-1}\sum_{\alpha}(-1)^{j}\langle\langle\tau_{2k_1-1-j}(\phi^\alpha)Q(\tau_j(\phi_\alpha))\rangle\rangle_{g-1}=0.
\end{align}

\begin{lemma}\label{lem:g00g}For any $l\geq1$, we have
\begin{align*}
&\frac{1}{2}\sum_{j=0}^{2l-1}\sum_{\alpha}(-1)^{j}\langle\langle\tau_{2l-1-j}(\phi^\alpha)\rangle\rangle_{g}\langle\langle Q(\tau_j(\phi_\alpha))\rangle\rangle_0
+\frac{1}{2}\sum_{j=0}^{2l-1}\sum_{\alpha}(-1)^{j}\langle\langle\tau_{2l-1-j}(\phi^\alpha)\rangle\rangle_{0}\langle\langle Q(\tau_j(\phi_\alpha))\rangle\rangle_g  
\\=&-l \sum_{j=0}^{2l-1}\sum_{\alpha}(-1)^{j}\langle\langle\tau_{j}(\phi_\alpha)\rangle\rangle_{0}\langle\langle \tau_{2l-1-j}(\phi^\alpha)\rangle\rangle_{g} 
+\sum_{j=0}^{2l-1}\sum_{\alpha}(-1)^{j}\langle\langle\tau_{2l-1-j}(\phi^\alpha)\rangle\rangle_{g}\langle\langle  Q(\tau_j(\phi_\alpha))\rangle\rangle_{0}.
\end{align*}
\end{lemma}
\begin{proof}
We only need to prove 
\begin{align*}
&\frac{1}{2}\sum_{j=0}^{2l-1}\sum_{\alpha}(-1)^{j}\langle\langle\tau_{2l-1-j}(\phi^\alpha)\rangle\rangle_{0}\langle\langle Q(\tau_j(\phi_\alpha))\rangle\rangle_g  
-\frac{1}{2}\sum_{j=0}^{2l-1}\sum_{\alpha}(-1)^{j}\langle\langle\tau_{2l-1-j}(\phi^\alpha)\rangle\rangle_{g}\langle\langle  Q(\tau_j(\phi_\alpha))\rangle\rangle_{0}
\\=&-l \sum_{j=0}^{2l-1}\sum_{\alpha}(-1)^{j}\langle\langle\tau_{j}(\phi_\alpha)\rangle\rangle_{0}\langle\langle \tau_{2l-1-j}(\phi^\alpha)\rangle\rangle_{g}. 
\end{align*}
By definition of $Q$ operator, the left hand side of the above equation equals to
\begin{align*}
&\frac{1}{2}\sum_{j,\alpha}(-1)^{j}(j+b_\alpha)\langle\langle\tau_{2l-1-j}(\phi^\alpha)\rangle\rangle_{0}\langle\langle \tau_j(\phi_\alpha)\rangle\rangle_g  
+\frac{1}{2}\sum_{j,\alpha}(-1)^{j}\langle\langle\tau_{2l-1-j}(\phi^\alpha)\rangle\rangle_{0}\langle\langle \tau_{j-1}(c_1(X)\cup\phi_\alpha)\rangle\rangle_g  
\\&-\frac{1}{2}\sum_{j,\alpha}(-1)^{j}(j+b_\alpha)\langle\langle\tau_{2l-1-j}(\phi^\alpha)\rangle\rangle_{g}\langle\langle  \tau_{j}(\phi_\alpha)\rangle\rangle_{0}
-\frac{1}{2}\sum_{j,\alpha}(-1)^{j}\langle\langle\tau_{2l-1-j}(\phi^\alpha)\rangle\rangle_{g}\langle\langle  \tau_{j-1}(c_1(X)\cup\phi_\alpha)\rangle\rangle_{0}
\\=& \frac{1}{2}\sum_{j,\alpha}(-1)^{j}(j+b_\alpha)\langle\langle\tau_{2l-1-j}(\phi^\alpha)\rangle\rangle_{0}\langle\langle \tau_j(\phi_\alpha)\rangle\rangle_g 
\\&+\frac{1}{2}\sum_{j,\alpha}(-1)^{j}(2l-1-j+b^\alpha)\langle\langle\tau_{j}(\phi_\alpha)\rangle\rangle_{g}\langle\langle  \tau_{2l-1-j}(\phi^\alpha)\rangle\rangle_{0}
\\&+\frac{1}{2}\sum_{j,\alpha}(-1)^{j}\langle\langle\tau_{2l-1-j}(\phi^\alpha)\rangle\rangle_{0}\langle\langle \tau_{j-1}(c_1(X)\cup\phi_\alpha)\rangle\rangle_g  
-\frac{1}{2}\sum_{j,\alpha}(-1)^{j}\langle\langle\tau_{2l-1-j}(c_1(X)\cup\phi_\alpha)\rangle\rangle_{g}\langle\langle  \tau_{j-1}(\phi^\alpha)\rangle\rangle_{0}
\\=&-l\sum_{j,\alpha}(-1)^{j}\langle\langle\tau_{j}(\phi_\alpha)\rangle\rangle_{0}\langle\langle \tau_{2l-1-j}(\phi^\alpha)\rangle\rangle_g.
\end{align*}
\end{proof}
By definition of vector fields $\mathcal{D}$ and $\mathcal{L}_0$ and Lemma~\ref{lem:g00g}, we obtain  the following formula for the left
hand side of equation~\eqref{eqn:sumnalphank1balpha} 
\begin{align*}
&LHS
\\=&k_1\langle\langle\tau_+^{2k_1}(-\mathcal{D})\rangle\rangle_g
+\langle\langle\tau_+^{2k_1}(\mathcal{L}_0)\rangle\rangle_g
+\sum_{\alpha}t_0^\alpha\langle\langle\tau_{2k-1}(c_1(X)\cup\phi_\alpha)\rangle\rangle_{g}
\\&+\frac{1}{2}\sum_{g_1+g_2=g}\sum_{j=0}^{2k_1-1}\sum_{\alpha}(-1)^{j}\langle\langle\tau_{2k_1-1-j}(\phi^\alpha)\rangle\rangle_{g_1}\langle\langle Q(\tau_j(\phi_\alpha))\rangle\rangle_{g_2}
\\&+\frac{1}{2}\sum_{j=0}^{2k_1-1}\sum_{\alpha}(-1)^{j}\langle\langle\tau_{2k_1-1-j}(\phi^\alpha)Q(\tau_j(\phi_\alpha))\rangle\rangle_{g-1}
\\=&k_1\langle\langle\tau_+^{2k_1}(-\mathcal{D})\rangle\rangle_g
+\langle\langle\tau_+^{2k_1}(\mathcal{L}_0)\rangle\rangle_g
+\sum_{\alpha}t_0^\alpha\langle\langle\tau_{2k-1}(c_1(X)\cup\phi_\alpha)\rangle\rangle_{g}
\\&+\frac{1}{2}\sum_{g_1+g_2=g,g_1,g_2>0}\sum_{j=0}^{2k_1-1}\sum_{\alpha}(-1)^{j}\langle\langle\tau_{2k_1-1-j}(\phi^\alpha)\rangle\rangle_{g_1}\langle\langle Q(\tau_j(\phi_\alpha))\rangle\rangle_{g_2}
\\&-k_1 \sum_{j=0}^{2k_1-1}\sum_{\alpha}(-1)^{j}\langle\langle\tau_{j}(\phi_\alpha)\rangle\rangle_{0}\langle\langle \tau_{2k_1-1-j}(\phi^\alpha)\rangle\rangle_{g} 
+\sum_{j=0}^{2k_1-1}\sum_{\alpha}(-1)^{j}\langle\langle\tau_{2k_1-1-j}(\phi^\alpha)\rangle\rangle_{g}\langle\langle  Q(\tau_j(\phi_\alpha))\rangle\rangle_{0}
\\&+\frac{1}{2}\sum_{j=0}^{2k_1-1}\sum_{\alpha}(-1)^{j}\langle\langle\tau_{2k_1-1-j}(\phi^\alpha)Q(\tau_j(\phi_\alpha))\rangle\rangle_{g-1}.
\end{align*}
Applying Lemma~\ref{lem:T2lL0genus-g} to the first three terms, we get 
\begin{align*}
LHS=&\langle\langle T^{2k_1}(\mathcal{L}_0-k_1\mathcal{D})\rangle\rangle_{g}
+\frac{1}{2}\sum_{g_1+g_2=g,g_1,g_2>0}\sum_{j=0}^{2k_1-1}\sum_{\alpha}(-1)^{j}\langle\langle\tau_{2k_1-1-j}(\phi^\alpha)\rangle\rangle_{g_1}\langle\langle Q(\tau_j(\phi_\alpha))\rangle\rangle_{g_2}    
\\&+\frac{1}{2}\sum_{j=0}^{2k_1-1}\sum_{\alpha}(-1)^{j}\langle\langle\tau_{2k_1-1-j}(\phi^\alpha)Q(\tau_j(\phi_\alpha))\rangle\rangle_{g-1}.
\end{align*}
By equation~\eqref{eqn:tau_jtau_m-j=TjTm-j}, we have
\begin{align*}
LHS=&\langle\langle T^{2k_1}(\mathcal{L}_0-k_1\mathcal{D})\rangle\rangle_{g}+\frac{1}{2}\sum_{g_1+g_2=g,g_1,g_2>0}\sum_{j=0}^{2k_1-1}\sum_{\alpha}(-1)^{j}\langle\langle T^{2k_1-1-j}(\phi^\alpha)\rangle\rangle_{g_1}\langle\langle Q(T^j(\phi_\alpha))\rangle\rangle_{g_2}
\\&+\frac{1}{2}\sum_{j=0}^{2k_1-1}\sum_{\alpha}(-1)^{j}\langle\langle T^{2k_1-1-j}(\phi^\alpha)Q(T^j(\phi_\alpha))\rangle\rangle_{g-1}. 
  \end{align*}
Recall the relations between $Q$ and $T$ operators
\begin{align*}
(QT^k-T^kQ)(W)=k T^k(W)-T^{k-1}(\mathcal{X}\bullet W)    
\end{align*}
for $k\geq1$.

Then the vanishing of $LHS$  follows from the fact that $\psi_1^{m}=0\,  (m>3g-2)$
on $\overline{\mathcal{M}}_{g,1}$ since the dimension of $\overline{\mathcal{M}}_{g,1}$ is $3g-2$ and $\psi_1^{i}\psi_2^{m-i}=0\, (m>3g-1)$ on $\overline{\mathcal{M}}_{g,2}$ since the dimension of $\overline{\mathcal{M}}_{g,2}$ is $3g-1$.

\end{document}